\documentclass[aps,pra,twocolumn,superscriptaddress,noshowpacs]{revtex4-2}
\usepackage[utf8]{inputenc}
\usepackage[T1]{fontenc}
\usepackage{graphicx}
\usepackage{xcolor}
\usepackage{physics}
\usepackage{bm}
\usepackage{amsfonts}
\usepackage{xurl}
\usepackage{nicefrac}
\usepackage{amsthm}
\usepackage{dsfont}
\usepackage{comment}
\usepackage[normalem]{ulem}

\newtheorem{theorem}{Theorem}[section]
\newtheorem{lemma}[theorem]{Lemma}
\newtheorem{observation}[theorem]{Observation}
\newtheorem{cor}[theorem]{Corollary}

\renewcommand{\ket}[1]{| #1 \rangle}
\renewcommand{\bra}[1]{\langle #1 |}
\renewcommand{\braket}[2]{\langle #1 | #2 \rangle}
\renewcommand{\ketbra}[2]{| #1 \rangle \langle #2 |}
\newcommand{\PP}{\mathcal{P}}
\newcommand{\rrp}{\mathfrak{r}_P}
\newcommand{\rr}{\mathfrak{r}}

\definecolor{dark-gray}{rgb}{.35,.55,.55}
\definecolor{dark-blue}{rgb}{.0,.0,.6}
\usepackage[colorlinks=true,linkcolor=dark-blue,citecolor=dark-blue,urlcolor=dark-blue]{hyperref}

\graphicspath{ {./pictures/} }

\begin{document}
\title{
Quantifying the dimensionality of multiparticle
entanglement via partition rank}

\author{Sophia Denker}
\email{sophia.denker@uni-siegen.de}
\affiliation{Naturwissenschaftlich-Technische Fakultät, Universität Siegen, Walter-Flex-Straße 3, 57068 Siegen, Germany}

\author{Ismaël Septembre}
\email{ismael.septembre@uni-siegen.de}
\affiliation{Naturwissenschaftlich-Technische Fakultät, Universität Siegen, Walter-Flex-Straße 3, 57068 Siegen, Germany}

\author{Robin Krebs}
\email{robin_benedikt.krebs@tu-darmstadt.de}
\affiliation{Department of Computer Science, Technische Universität Darmstadt, Germany}

\author{Otfried Gühne}
\email{otfried.guehne@uni-siegen.de}
\affiliation{Naturwissenschaftlich-Technische Fakultät, Universität Siegen, Walter-Flex-Straße 3, 57068 Siegen, Germany}

\date{\today}
\begin{abstract}
The usefulness of entanglement as a resource in quantum technologies increases for 
larger systems, that is, if more particles or higher-dimensional quantum systems 
are considered. Yet, the interplay between dimensionality and multiparticle entanglement
is not well understood. Only for two-particle systems an unambiguous and coherent 
notion of entanglement dimensionality, based on the Schmidt decomposition, is known. 
We introduce a concept to characterize the entanglement dimensionality of multiparticle 
states based on decompositions of pure states into superpositions of states without genuine
multiparticle entanglement. We provide constructive methods to characterize the resulting
partition rank for pure and mixed states. This allows the
identification of novel maximally correlated states as well as a discrete classification 
of quantum states under stochastic local operations and classical communication. From
a mathematical perspective, our approach can be formulated in terms of
the slice rank and 
partition rank of tensors, and our results allow to characterize these by connecting
them to a generalized injective tensor norm. 

\end{abstract}
\maketitle

%%%%%%%%%%%%%%%%%%%%%%%%%%%%%%%%%%%%%%%%%%%%%%%%%%%%%%%%%%%%%%%%
{\it Introduction.---} 
%%%%%%%%%%%%%%%%%%%%%%%%%%%%%%%%%%%%%%%%%%%%%%%%%%%%%%%%%%%%%%%%
Entanglement is a central phenomenon of quantum physics through which 
the states of spatially separated particles can be interdependent~\cite{Horodecki2009}. 
It leads to advantages in quantum metrology \cite{TothApellaniz2014,PezzeSmerzietal2018}
 or quantum communication~\cite{gisin2007quantum,luo2023recent}. In the latter task, 
the usage of high-dimensional systems instead of qubits leads to additional
advantages~\cite{cozzolino2019high,CerfBourennaneKarlsson2002}. 
Characterising entanglement for two qubits is straightforward and many 
criteria exist~\cite{ekert1995entangled, peres1996separability, HORODECKI19961, wootters2001entanglement, nielsen2010quantum}; some also generalise to 
qudits. Communication scenarios, however often involve more 
than two parties, which complicates entanglement
characterization~\cite{walter2016multipartite, GuehneToth2009}. For 
instance, already for three qubits, inequivalent forms of entanglement 
exist \cite{coffman2000distributed,dur2000three}, and the
situation becomes more complicated for higher dimensions~\cite{chitambar2010matrix}.
This leads to a central question: How can we characterize the dimensionality
of multiparticle quantum systems?

For two-particle systems, the Schmidt decomposition delivers a natural tool to characterize the dimensionality of entanglement
\cite{Terhal2000, SanperaBrussLewenstein2001}. This calls for a straightforward extension to the multiparticle
case: One can partition the particles into two groups and then 
consider each group as a high-dimensional quantum system to recover 
the standard bipartite configuration that is still tractable~\cite{nielsen2010quantum}. Since there are many ways to partition a multiparticle system, the Schmidt decomposition in each partition gives information about the global system, leading to a vector of bipartite dimension parameters~\cite{HuberdeVicente2013}. Then, if no partition is exempt from entanglement, the state is genuinely multipartite entangled, and the actual number of entangled levels can be quantitatively assessed~\cite{CobucciTavakoli2024}. 

Nevertheless, all these constructions rely on fixed bipartite decompositions, so one may doubt that they can capture the essence of multiparticle entanglement. A different way to
extend the gist of the Schmidt decomposition to multipartite
systems is to expand multiparticle states as superpositions of fully separable states \cite{EisertBriegel2001}. Still, this is 
analytically difficult \cite{haastad1990tensor,fawzi2022discovering} and, in addition, this does not discriminate between the notions of biseparability and genuine multiparticle entanglement.

In this work, we define a truly multipartite measure that
naturally unveils genuine multipartite entanglement properties 
of quantum states. It is based on the number of product terms of arbitrary partitions needed to decompose a pure state. Following the mathematical literature, which has started to investigate an analogous concept for tensors \cite{tao2016capset, tao2016slice, Naslund2020}, we call it the \textit{partition rank} of a quantum state.
We show that the partition rank reveals stark differences in entanglement properties invisible to earlier methods. For that
we present methods to study the partition rank in practice 
for pure and mixed states by reformulating it as a geometric distance minimization
problem~\cite{wei2003geometric,WeinbrennerGuehne2025}, enabling 
the use of methods such as see-saw algorithms and semi-definite 
programming. We then identify relevant quantum states which are maximally entangled with respect those measures. 

%%%%%%%%%%%%%%%%%%%%%%%%%%%%%%%%%%%%%%%%%%%%%%%%%%%%%%%%%
%%%%%%%%%%%%%%%%%%%%%%%%%%%%%%%%%%%%%%%%%%%%%%%%%%%%%%%%%

%%%%%%%%%%%%%%%%%%%%%%%%%%%%%%%%%%%%%%%%%%%%%%%%%%%%%%%
{\it Motivation and main idea.---}
%%%%%%%%%%%%%%%%%%%%%%%%%%%%%%%%%%%%%%%%%%%%%%%%%%%%%%%
To start, recall that for a pure bipartite state $\ket{\psi}$ 
the Schmidt decomposition is given by~\cite{schmidt1907theorie,nielsen2010quantum}
\begin{align}
    \ket{\psi} = \sum_{i=1}^{\rr}s_i\ket{a_i}_A \ket{b_i}_B,
\end{align}
where the minimal number $\rr$ of nonzero Schmidt coefficients $s_i$ 
is called the \textit{Schmidt rank}. This effectively counts 
how many product terms are needed to decompose a state and 
therefore quantifies the dimensionality of bipartite entanglement. Indeed, separable states have Schmidt rank 
one, while any state that is entangled in some level 
has a larger Schmidt rank. In the standard formulation
the Schmidt vectors $\ket{a_i}_A$ (and $\ket{b_i}_B$) form 
orthogonal bases and the Schmidt coefficients $s_i$ are real 
and positive. For later generalization, however, it is useful
to note that the requirements of orthogonality and positivity are
not essential --- dropping them leads to the same notion of a 
Schmidt rank $\rr$. 

How can this approach be generalized to multiparticle systems? 
First, in this case the entanglement structure becomes  more intricate~\cite{GuehneToth2009,Horodecki2009,walter2016multipartite}. For instance, for three particles one can distinguish 
between \textit{fully separable} states, which can be 
written as $\ket{a}\ket{b}\ket{c} \equiv\ket{abc}$; 
biseparable states, which are of the form $\ket{a}_A\ket{\phi}_{BC}$, $\ket{b}_B\ket{\phi}_{AC}$, or $\ket{c}_C\ket{\phi}_{AB}$, and states that do 
not have any product structure. The latter ones are called \textit{genuine multipartite entangled} (GME). 

Then, we can ask the same question as in the bipartite case: How many product terms are needed to decompose a multipartite state? Here, however, there are more product structures possible than only the partition $A|B$ corresponding to $\ket{ab}$. For three particles, we can identify a set of possible nontrivial partitions, e.g.,
$\mathcal{P} =\{A|B|C, A|BC, B|AC, C|AB\}$, where 
the first partition corresponds to a fully separable term and the others 
to biseparable terms. 
Then we can define:

\textit{Given a multiparticle state $\ket{\psi}$ and a 
set of partitions $\mathcal{P}$, one can consider decompositions
\begin{align}
    \ket{\psi} = \sum_{i=1}^{r_\mathcal{P}}
    \ket{\nu_i^P}
    \label{eq:partitionrank}
\end{align}
where the non-normalized states $\ket{\nu^P_i}$ are each separable with 
respect to at least one partition $P\in \mathcal{P}$. The smallest possible 
number of terms $r_\mathcal{P}$ in this decomposition is called the partition
rank (with respect to the set $\mathcal{P}$).} 

The partition rank quantifies the minimal dimension that is needed, if the state $\ket{\psi}$ should be embedded in a space spanned by vectors without genuine multipartite entanglement. Consequently, it can be considered as a measure for the dimensionality of multipartite entanglement.

Some remarks are in order. First, for many cases like the example set $\PP$ above, it is sufficient to consider bipartitions only, since any 
term with a finer product structure can just be assigned to a compatible bipartition. Explicitly, in the three-particle case, we can consider decompositions of the form $\ket{\psi} = \sum_i\ket{\varphi_i}_A\ket{\phi_i}_{BC}+\sum_{j>i}\ket{\varphi_j}_B\ket{\phi_j}_{AC}+\sum_{k>j}\ket{\varphi_k}_C\ket{\phi_k}_{AB}$. 
In fact, for the three-particle case the states
in Eq.~(\ref{eq:partitionrank}) can be chosen to be mutually orthogonal, see Appendix~\ref{app:orth}. This 
highlights the analogy to the bipartite Schmidt 
decomposition and the corresponding notions of dimensionality. 

Second, when restricting the set of allowed partitions $\mathcal{P}$ 
to the fully separable one, e.g., $A|B|C$ in the three-partite case, we 
recover the decompositions connected to the Schmidt measure~\cite{EisertBriegel2001}, while restricting it to one 
fixed bipartition, lets say $\mathcal{P}=\{A|BC\}$, corresponds 
to the Schmidt decomposition. The advantage of our approach lies
in the fact that the notion of partition rank can, contrary to 
the Schmidt measure, quantify genuine multipartite entanglement and, 
at the same time, goes beyond the consideration of a single bipartition.

Finally, a similar concepts have been studied in mathematics~\cite{KlinglerNetzerDelesCoves2025border}. There, the Schmidt measure corresponds to the notion of  tensor rank \cite{Landsberg2011TensorsGA}. 
If one considers only the bipartitions of one particle vs.~the rest, then one arrives at the notion of slice rank \cite{tao2016capset, tao2016slice}, and the general notion from Eq.~(\ref{eq:partitionrank}) has recently been studied independently as partition rank \cite{Naslund2020, karam2023smallsunflowersstructureslice, CohenMoshkovitz2023, BikDraismaLampertZiegler2025partiitonrank, Oneto2025review},
where we also borrow the
name from. In the following, we will see how the physical interpretation
helps to gain insights into these mathematical constructions.

%%%%%%%%%%%%%%%%%%%%%%%%%%%%%%%%%%%%%%%%%%%%%%%%%%%%%%%
{\it Relation to other measures.---}
%%%%%%%%%%%%%%%%%%%%%%%%%%%%%%%%%%%%%%%%%%%%%%%%%%%%%%%
To start, we show that our measure goes beyond the Schmidt-number vector and the GME-dimension known from the literature \cite{HuberdeVicente2013, CobucciTavakoli2024}. Consider the three-qutrit state
\begin{equation}
\label{eq:state2bisep}
\ket{\psi_2} = 
\frac{1}{\sqrt{2}}
\big( 
\ket{0}_A \ket{\Phi_{0,1}^+}_{BC}  
+ \ket{\Phi_{1,2}^+}_{AB} 
\ket{2}_C 
\big),
\end{equation}
where $\ket{\Phi _{j,k}^+} \propto \sum_{m=j}^k\ket{mm}$ denotes a 
maximally entangled state on the levels $j$ to $k$. By direct calculation, 
we obtain Schmidt rank $r=3$ for every bipartition, so the Schmidt-number
vector is $r_\mathrm{SN}=(3,3,3)$ and the GME-dimension as the minimum
of the bipartite Schmidt ranks is $d_\mathrm{GME}=3$. However, $\ket{\psi_2}$ 
is a sum of only two biseparable states, meaning that the partition rank is $r_\mathcal{P}=2$. 

Going further, the gap between the partition rank and the GME-dimension 
can be arbitrarily large. If one takes the three-qudit state
$\ket{\psi^d_2} =
(\ket 0_A \ket{\Phi _{0,d-1}^+}_{BC}  + \ket{\Phi _{1,d}^+}_{AB} \ket{2}_C 
)/{\sqrt{2}}$, its partition rank is $r_\mathcal{P}=2$ for all 
$d \geq 2$ but its GME-dimension is $d_\mathrm{GME}=d$.

In the other direction, since the GME-dimension is based on the Schmidt decomposition, which is a special case of the partiton rank decomposition, 
we observe: 

\textbf{Observation 1.}
\textit{For any pure state $\ket{\psi}$ the partition rank is a lower 
bound on the GME-dimension
\begin{align}
    r_\mathcal{P}\leq d_\mathrm{GME}.
\end{align}}

Indeed, there are many states with $d_\mathrm{GME} = r_\mathcal{P}$ 
such as for example the three-particle Greenberger-Horne-Zeilinger (GHZ)
state, which for qutrits is given by 
$\ket{\mathrm{GHZ}}=(\ket{000}+\ket{111}+\ket{222})/{\sqrt{3}}$. 

%%%%%%%%%%%%%%%%%%%%%%%%%%%%%%%%%%%%%%%%%%%%%%%%%%%%%%%%%%%
{\it Entanglement quantifier based on partition rank.---}
%%%%%%%%%%%%%%%%%%%%%%%%%%%%%%%%%%%%%%%%%%%%%%%%%%%%%%%%%%%
It is desirable to quantify how well a state can be approximated 
with a state of smaller partition rank $k<r_\mathcal{P}$. Similar 
to the concept of the geometric measure of entanglement 
\cite{wei2003geometric, WeinbrennerGuehne2025} or its 
generalizations \cite{SenDeSen20210gengm, Prabhuetal2012gengm, ZhuZhangZeng2024subspaceborderrank}, where the aim is to find 
the distance of a given state to the set of fully separable states, we 
can introduce an entanglement quantifier based on partition rank. 

\textit{Given a pure multipartite state $\ket{\psi}$ consider the distance
\begin{align}
\Omega_k^2(\psi) = \sup_{\ket{\eta} \in \mathbb{P}_k} |\braket{\eta}{\psi}|^2,
\end{align}
where $\mathbb{P}_k$ is the set of all states with partition rank 
$r_\PP \leq k$.  Then, we call $G_k(\psi)=1-\Omega_k^2(\psi)$ the 
geometric measure of partition rank $k$.}

Moreover, determining this overlap $\Omega_k^2$ allows us to 
construct witnesses for partition rank $k$ states. The logic 
is the same as in the case of entanglement witnesses \cite{Bourenanneetal2004fidwit}. States in the vicinity of 
entangled states are entangled, too and hence the witness 
\begin{equation}
    \mathcal{W}_k = \Omega^2_k \mathds{1} - \ketbra{\psi}{\psi}
\end{equation} will detect states close to $\ket{\psi}$ which 
have partition rank at least $k+1$.

The key problem is to compute the overlap $\Omega^2$, where we omit the index $k$ for better readability. Interestingly, 
as we will show below, this can be effectively seen as a geometric 
measure of entanglement in an extended space, i.~e., it falls back 
to considering the trivial partition $A|B$. The advantage is that
the geometric measure of entanglement is well-studied and many 
tools for its computation are known \cite{peres1996separability, 
HORODECKI19961, GuehneReimpellWerner2007seesaw, GerkeVogelSperling2018seesaw, Weinbrenneretal2026hier}. Since our construction is quite general, this 
also demonstrates how methods from entanglement theory can be employed 
to gain insights into  mathematical problems related to the partion rank
\cite{Naslund2020, karam2023smallsunflowersstructureslice, CohenMoshkovitz2023, BikDraismaLampertZiegler2025partiitonrank, Oneto2025review}.

We start by considering the approximation of an arbitrary three-qutrit 
state $\ket{\psi}$ with a state $\ket{\eta}$ of partition rank $k = 2$,
the method can directly be generalized to other cases. Concretely, we 
take states of the form $\ket{\eta} = \ket{\varphi_1}_A\ket{\phi_1}_{BC}+\ket{\phi_2}_{AB}\ket{\varphi_2}_C$, where $\ket{\eta}$
is normalized, but the one-qutrit states $\ket{\varphi_j} \in \mathbb{C}^3$ 
and the two-qutrit states $\ket{\phi_j} \in \mathbb{C}^9$ are not necessarily normalized, and we wish to maximize $|\braket{\psi}{\eta}|^2$.
If $\ket{\eta}$ is a superposition of two other different bipartitions, 
the method below can readily be adapted, and if $\ket{\eta}$ is a 
superposition of two states which are separable to the same bipartition, 
then one just needs to maximize the overlap with states of Schmidt rank 
two, which is given by the the sum of the squares of the largest two 
Schmidt coefficients  of $\ket{\psi}$ \cite{Guehne2004}.

We arrange the states $\ket{\varphi_j}$ (and $\ket{\phi_j}$) corresponding 
to the smaller (larger) systems in two vectors by taking the direct sum of 
the respective states: $\ket{s} = \ket{\varphi_1}_A \oplus \ket{\varphi_2}_C \in \mathbb{C}^6$ and $\ket{L} = \ket{\phi_1}_{BC} \oplus \ket{\phi_2}_{AB}\in \mathbb{C}^{18}$. This extension (and the map back) can be described by a canonical operator $E$ fulfilling $\ket{\eta} = E \ket{s}\ket{L}$, in fact, 
$E$ maps from $\mathbb{C}^6 \otimes \mathbb{C}^{18}$ to the three-qutrit 
space $(\mathbb{C}^3)^{\otimes 3}.$ The optimization problem then translates 
to (see  Appendix~\ref{app:algo} for more details)
\begin{align}
    \Omega^2(\psi) 
    &= \sup_{\ket{\eta}\in \mathbb{P}_2} |\braket{\psi}{\eta}|^2
    \label{eq:translation}
    \\
    &= \sup_{\ket{s}, \ket{L}}|\langle \Psi \ket{s L}|^2, \text{ such that } \bra{sL}E^\dagger E\ket{sL} = 1, 
    \nonumber
\end{align}
where $\ket{\Psi} = E^\dagger \ket{\psi}$ corresponds to the state $\ket{\psi}$
lifted by the adjoint $ E^\dagger$ to the high-dimensional space
and the normalization of the state $\ket{\eta}$ is encoded in 
the matrix $M = E^\dagger E$. In fact, for $M=\mathds{1}$ this
would reduce to the bipartite geometric measure of entanglement, demonstrating 
the analogy explicitly.

We can therefore adapt two well-established tools to lower- and upper-bound 
the overlap $\Omega^2$. First, one can generalize see-saw algorithms \cite{GuehneReimpellWerner2007seesaw, GerkeVogelSperling2018seesaw} to
find good approximations $\ket{sL}$ and consequently find a good lower bound
on $\Omega^2$. For that, note that for a fixed $\ket{s}$ in Eq.~(\ref{eq:translation}) the optimal $\ket{L}$ can analytically 
be determined as solution of a generalized eigenvalue problem. 
This allows for an iteration, by starting
with a random $\ket{s}$, and first optimizing $\ket{L}$ and then 
optimizing $\ket{s}$ etc.

Second, one can use relaxations of the optimization problem
to obtain upper bounds. The set of product states $\sigma=\ketbra{sL}{sL}$ 
with the appropriate normalization $\tr(\sigma M) = 1$ is a subset of 
the appropriately normalized mixed quantum states fulfilling the 
criterion of the positivity of the partial transpose (PPT) \cite{peres1996separability,HORODECKI19961}. Maximization over this subset
can reliably be done via semidefinite programs \cite{reviewSDP2024,WeinbrennerGuehne2025}, in fact, convergent hierarchies leading to an exact approximation can be derived~\cite{DPShierarchy2004,berta2022semidefinite}.
The combination of both methods leads to upper and lower bounds
$
(\Omega^2)^\mathrm{s-s} \leq \Omega^2 \leq (\Omega^2)^\mathrm{SDP},
$
details of the methods are explained in Appendix~\ref{app:algo}.

These two methods work particularly well when looking for approximations
with low partition rank, but the optimization becomes rather large in 
system size when going to higher partition rank approximations. To that 
end, we introduce a third algorithm based on the Gauß-Newton method. The key idea here is the following: 
Consider the state 
$\ket{\psi} = \sum_{P\in \mathcal{P}}\ket{\gamma_P}$ written in some
(non-optimal) decomposition, which should be approximated by $\ket{\eta}.$ 
Each term 
$\ket{\gamma_P}$
can be written in its Schmidt decomposition with respect to the 
partition $P \in \mathcal{P}$. 
In order to bring $\ket{\psi}$ in the form of $\ket{\eta}$ above, we 
aim for $\ket{\gamma_{A|BC}}$ and $\ket{\gamma_{C|AB}}$ to have only 
one nonzero Schmidt coefficient each and for $\ket{\gamma_{B|AC}}$ 
to vanish completely. We arrange the Schmidt coefficients, which 
are supposed to be zero in a vector and minimize them with a Gauß-Newton 
algorithm (more details in Appendix~\ref{app:algo}). Then, if the algorithm 
converges to the zero vector, we will find an exact decomposition 
of the form $\ket{\eta}$, otherwise we obtain a good approximation and
a lower bound on $\Omega^2$ {by truncating} the Schmidt decompositions.
Note that a similar method was introduced in Ref.~\cite{TomiokaSuzuki2013MLpaper}. 

%%%%%%%%%%%%%%%%%%%%%%%%%%%%%%%%%%%%%%%%%%%%%%%%%%%%%
{\it Application to various three-particle states.---}
%%%%%%%%%%%%%%%%%%%%%%%%%%%%%%%%%%%%%%%%%%%%%%%%%%%%
When determining the geometric measure of partition rank two, different combinations of partitions have
to be taken into account. If both terms in the superposition are separable for the same partition, 
the maximal overlap is, as mentioned above, given 
by the sum of the two largest Schmidt coefficients squared. Only for the case that they are separable for different partitions the algorithms introduced above
are needed. For our computations, we focus on the 
lower bound from the see-saw algorithm and the upper bound coming from the PPT relaxation; the results 
of the Gauß-Newton-based algorithm coincide with those of the see-saw. We then compute the geometric measure of partition rank 2, following $G_2(\psi) = 1-\sup_{\ket{\eta}\in \mathbb{P}_2}|\braket{\eta}{\psi}|^2$, and summarize the results in Table~\ref{tab:results}.

\begin{table}[t]
    \centering
    \begin{tabular}{|c|c|c|c|c|}\toprule
       state  & $\max_{\rm bp} (s_1^2+s_2^2)$ & $(\Omega^2)^{\mathrm{s-s}}$ & $(\Omega^2)^{\mathrm{SDP}}$ & $G_{2}$ \\\hline
       $\ket{\mathrm{GHZ}}$  & $\nicefrac{2}{3}$ & $\nicefrac{2}{3}$ & $\nicefrac{2}{3}$ & $\nicefrac{1}{3}$\\
       $\ket{\mathrm{S}}$  & $\nicefrac{2}{3}$ & $\nicefrac{2}{3}$ & $\nicefrac{2}{3}$ &  $\nicefrac{1}{3}$\\
       $\ket{\psi_2}$ & $\nicefrac{3}{4}$ & 1 & 1 & 0\\
       $\ket{\psi(1,1,1)}$ & $\nicefrac{5}{6}\approx 0.83$ & 0.7968  & 0.8058  & $\nicefrac{1}{6}$\\
       $\ket{\psi(3,4,5)}$ & $\nicefrac{41}{50}=0.82$ & 0.8424  &  0.8780  & $\geq$ 0.1220\\\hline
    \end{tabular}
    \caption{Overlaps $\Omega^2$ with states of partition rank two for different target states (defined in the text) and computed with different methods. Second column: If both terms in the superposition are separable for the same partition, the overlap
is given by the sum of the two largest squared Schmidt coefficients. Third and fourth column: If both terms belong to different bipartitions, the maximal $(\Omega^2)^{\rm{s-s}}$ and $(\Omega^2)^{\text{SDP}}$ 
give lower and upper bounds. Fifth column: 
The lower bound on the geometric measure of 
partition rank two is computed by $G_2\geq 1-\max\{s_1^2+s_2^2,(\Omega^2)^{\rm SDP}\}$. See text for further details.}
    \label{tab:results}
\end{table}

From these examples we see that different entanglement structures are possible. First, the GHZ state as well as the three-qutrit supersinglet state $\ket{\mathrm{S}}=\nicefrac{1}{\sqrt{6}}(\ket{012}+\ket{201}+\ket{120}-\ket{102}-\ket{210}-\ket{021})$ are both equally well approximated by states of 
Schmidt rank 2 (for a fixed bipartition)
and states partition rank 2 (containing two different bipartions). Second, the state $\ket{\psi_2}$ from Eq.~(\ref{eq:state2bisep}), has partition rank two but its largest overlap with states of Schmidt rank two is only $\nicefrac{3}{4}$. Third, we introduce the state
\begin{align}
    \ket{\psi(x,y,z)}&=\mathcal{N}[x(\ket{210}+\ket{201})\\\nonumber
    &+y(\ket{100}+\ket{122})+z(\ket{011}+\ket{022})],
\end{align}
where $\mathcal{N}$ denotes the normalization. For $x=y=z=1$ the state has partition rank three and its best rank-2 approximation is a Schmidt rank-2 state. For $x=3$, $y=4$, and $z=5$, however, it is better approximated by a partition-rank-2 state with mixed bipartitions. We further see that in the first examples, the results obtained by the SPD coincide with the ones obtained by the see-saw. In these cases, the result is tight, while for the states $\ket{\psi(x,y,z)}$ we can only bound the overlap.

The results can be used to construct witnesses for partition rank three as mentioned above. This distinguishes the partition rank once more from the GME-dimension, since we see that for the state $\ket{\psi(3,4,5)}$ a fidelity exceeding 0.820 would certify GME-dimension 3, but to certify true partition rank three 
definitely a fidelity larger than $0.8424$
%0.878 
is needed, as the see-saw algorithm finds partition rank-two states with this fidelity.

%%%%%%%%%%%%%%%%%%%%%%%%%%%%%%%%%%%%%%%%%%%%%%%%
{\it Identifying maximally entangled states.---} 
%%%%%%%%%%%%%%%%%%%%%%%%%%%%%%%%%%%%%%%%%
Having defined the entanglement quantifier,
we can ask for the maximally entangled state with respect to the partition rank. Precisely, we aim to identify the state $\ket{\psi}$ and the minimal value $\Omega^2_{\mathrm{opt}}$ such that
\begin{align}
    \Omega^2_\mathrm{opt} = \min_{\ket{\psi}}\sup_{\ket{\eta}\in \mathbb{P}_2}|\braket{\psi}{\eta}|^2.
\end{align}
To that end, we use the algorithm described in Ref.~\cite{SteinbergGuehne2024} and replace the part of finding the closest product state by finding the closest partition rank 2 state. We find that the algorithm always converges to a state whose marginals are all maximally mixed, reaching a value of $\Omega^2_\mathrm{opt} = \nicefrac{2}{d}$ for $d = 3,4$. Such states are called absolutely maximally entangled (AME)~\cite{scott2004ame,HuberGuehneSiewert2017ame,Ratheretal2022ame,Rajchel_Mieldzio__2026AME}. 

Using the fact that for three particles the smallest possible overlap with Schmidt rank 2 states is $\nicefrac{2}{d}$ and is only reached by AME states. We can give an analytical proof (details in Appendix~\ref{app:AME}) and generalize our observation.

\textbf{Observation 2.} \textit{For three-particle states, the geometric measure of partition rank two is maximized by $\ket{\psi}$ if and only if $\ket{\psi}$ is absolutely maximally entangled. We then have $\Omega^2 = \sup_{\eta\in\mathbb{P}_2}|\braket{\eta}{\psi}|^2 = \nicefrac{2}{d}$.}

Note that for more particles, the situation becomes more intricate. First, the product structure is more variable, so the simple argument from before might not hold. Second AME states do not exist for all combinations of particle number $N$ and local dimension $d$ immediately excluding a direct generalization to $N$ particles. We can, however, derive also bounds on the geometric measure of slice rank for general states based on the Schmidt coefficients for all bipartitions, see Appendix~\ref{app:slice} for details.

{\it The partition rank vector.---}
So far, the discussion concerned the number of terms in a decomposition as in Eq.~(\ref{eq:partitionrank}), but not their detailed structure. In order to arrive at a more sensitive classification, we may count the terms
for the different bipartitions, introducing the partition rank vector (PRV).

\textit{Consider a state $\ket{\psi}$ where one orders a 
general, not necessarily optimal, decomposition from Eq.~(\ref{eq:partitionrank})
as
\begin{align}
\ket{\psi} = \sum_{P\in \mathcal{P}} \sum_{\ell=1}^{\rr_P} \ket{\mu^P_\ell},
\label{eq-prv}
\end{align}
combining states $\ket{\mu^P_\ell}$ which are separable for the same bipartition $P$. This means that
we are considering superpositions of states 
$\ket{\gamma_P}=\sum_{\ell=1}^{r_P} \ket{\mu^P_\ell}$
with Schmidt rank $\rr_P$
for the given bipartition $P$. Then, we arrange the $\rr_P$ in the partition rank vector (PRV)
$\mathbf{r}$. For three particles, we have 
$\mathbf{r} = (\rr_{A|BC},\rr_{B|AC},\rr_{C|AB})$.
}

A given state $\ket{\psi}$ may admit several decompositions, corresponding to different 
PRVs, not all may correspond
to the partition rank decomposition as in Eq.~\eqref{eq:partitionrank}.
In fact, for three particles we can use
the orthogonalizability of the terms in 
Eq.~(\ref{eq-prv}) to derive criteria to 
exclude certain PRVs from the bipartite Schmidt coefficients,  see Appendix~\ref{app:orth} for a full discussion. 

Moreover, 
this gives rise to
two definitions: First, we call a PRV \textit{tight} if every entry is minimal, meaning
that there is no decomposition of $\ket{\psi}$
with a PRV where one entry is smaller and the others
are the same. For example if a state with PRV 
$\mathbf{r} = (2,2,2)$ cannot be decomposed 
following $\mathbf{r}' =(1,2,2)$ or permutations thereof, the decomposition assigned to $\mathbf{r}$ is tight. 
Second, we call a PRV \textit{optimal}, if it corresponds
to an optimal decomposition as in Eq.~\eqref{eq:partitionrank}, meaning 
that the sum of the entries $\rrp$ equals the partition rank $r_\PP$.
Clearly, any optimal PRV is also tight but the converse does 
not hold. For an example and more details see 
Appendix~\ref{app:dec}.
 
In order to demonstrate that the PRV leads to new insights,  
let us consider for three qutrits the state family 
\begin{align}
    \ket{\psi(\theta)} = \sin(\theta)\ket{\mathrm{GHZ}} + \cos(\theta)\ket{\mathrm{S}}, \label{eq:ghzsparam}
\end{align}
for $\theta\in[0,\nicefrac{\pi}{2}]$. All of these states are 
AME and therefore have Schmidt-number vector~\cite{HuberdeVicente2013} $r_\mathrm{SN} = (3,3,3)$, GME-dimension~\cite{CobucciTavakoli2024} $d_\mathrm{GME} = 3$, partition rank $r_\mathcal{P} = 3$. Moreover, 
they all maximize the geometric measure of partition rank two with $G_2(\psi(\theta))=\nicefrac{2}{3}$. But as we will see now, the PRV
leads to a finer classification.

First, we note that since all states have Schmidt-number 
vector $r_\mathrm{SN} = (3,3,3)$, they all admit decompositions 
with the PRV $\mathbf{r} = (3,0,0)$ and permutations, which are 
optimal. For partition rank three, two other potential optimal 
PRVs are $\mathbf{r} = (2,1,0)$ (and permutations) and 
$\mathbf{r} = (1,1,1)$.

The GHZ state $\ket{\psi(\theta=\nicefrac{\pi}{2})}$ admits all 
of them. 
Contrary to that, the singlet state $\ket{\mathrm{S}} = \ket{\psi(\theta=0)}$, 
does not admit decompositions with the PRV $\mathbf{r} = (2,1,0)$
and permutations, as can be seen from  Fig.~\ref{fig:ghzsparam} and the discussion below. However, one can easily check that it admits, decompositions with the PRV
$\mathbf{r} = (2,2,0)$, which is therefore tight. 

We now go further and, similar to above, consider the largest overlap 
with states which admit a PRV,
$\Omega^2_\mathbf{r} = \sup_{\ket{\eta} \in \mathbf{r}}
|\braket{\eta}{\psi}|^2$. This can be estimated with the algorithms
mentioned above. In Fig.~\ref{fig:ghzsparam} we plot the overlaps $\Omega^2_{(1,1,1)}$ and $\Omega^2_{(2,1,0)}$ parametrized by 
the angle $\theta$.

\begin{figure}
    \centering
    \includegraphics[width=\linewidth]{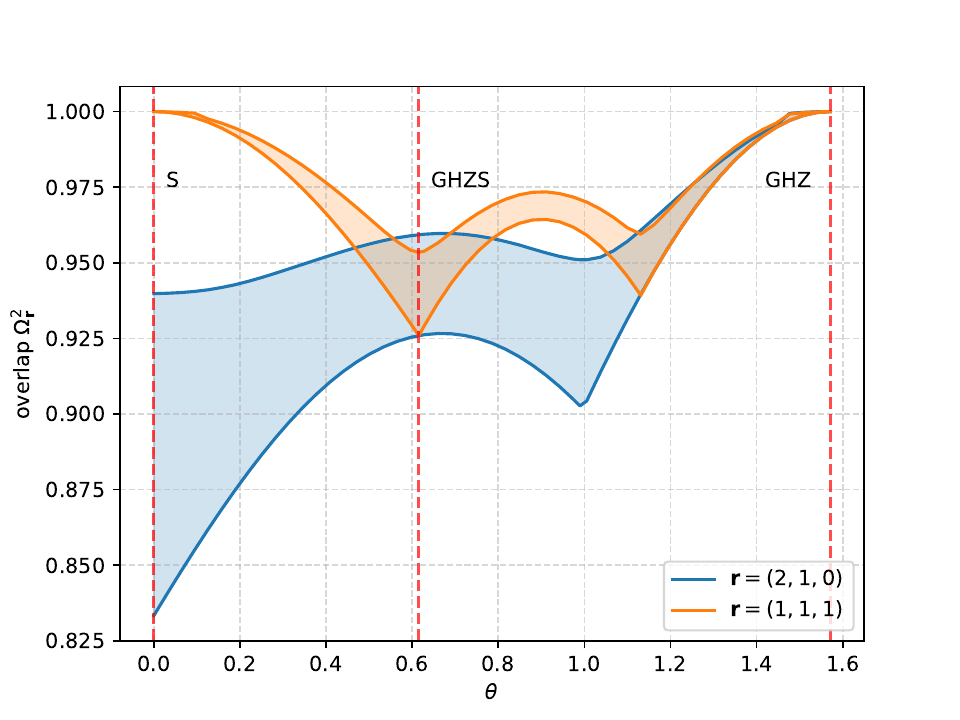}
    \caption{Overlap $\Omega^2_{\mathbf{r}}$ 
    for different parametrizations of the state $\ket{\psi}=\cos(\theta)\ket{\mathrm{S}}+\sin(\theta)\ket{\mathrm{GHZ}}$. 
    The upper bounds were computed by the PPT relaxation, the lower bounds 
    by the see-saw algorithm. We find that in this case the vector $\mathbf{r}=(2,1,0)$ includes all permutations. }
    \label{fig:ghzsparam}
\end{figure}

We note that almost all states $\ket{\psi(\theta)}$ with $\theta \neq 0$ and $\theta \neq \nicefrac{\pi}{2}$ do not allow the PRV $\mathbf{r} = (2,1,0)$ 
and $\mathbf{r} = (1,1,1)$ and therefore admit tight vectors that are not optimal. Moreover, we see that for $\mathbf{r} = (2,1,0)$ the state with the largest geometrical distance, i.e. the smallest overlap $\Omega^2_{\mathbf{r}}$ is the singlet state, reaching the value $\nicefrac{5}{6}\leq \Omega^2_{(2,1,0)}\leq 0.9398$. For the vector $\mathbf{r} = (1,1,1)$ 
the state with largest distance seems to be given by
\begin{align}
    \ket{\mathrm{GHZS}}=\frac{1}{\sqrt{3}} \big( |\mathrm{GHZ}\rangle + \sqrt{2}|\mathrm{S}\big),
\end{align}
with $\nicefrac{25}{27}\leq \Omega^2_{(1,1,1)}\leq 0.9535$ (see also Fig.~\ref{fig:ghzsparam}).
Independently of that state family we also looked systematically for states minimizing the overlap $\Omega^2_\mathbf{r}$, using the adapted algorithm from Ref.~\cite{SteinbergGuehne2024}. We then find that the algorithm converges to %these states 
the state $\ket{\mathrm{GHZS}}$ as well. 
Finally, that the GHZS state has no decomposition with PRV $\mathbf{r}=(1,1,1)$ can also be shown with an algebraic certification (see Appendix~\ref{app:cert}).

{\it Partition rank vectors and state transformations.---}
The PRV has an operational meaning as it can characterize the impossibility
of transformations between single copies of states using local operations and classical communication (SLOCC) \cite{dur2000three}. 
Recall that a state $\ket{\phi}$
can be prepared from $\ket{\psi}$ via SLOCC~\cite{vidal1999entanglement}, if there are matrices $A,B,C$ such that $\ket{\phi} = A \otimes B \otimes C \ket{\psi}.$ 
Such transformations can only preserve
or reduce the number of terms in decompositions as Eqs.~(\ref{eq:partitionrank}, \ref{eq-prv}) \cite{chitambar2008tripartite}. For capturing this effect, we define a relation between PRVs as follows.
We say $\mathbf{r} < \mathbf{r}'$, if the vector $\mathbf{r}' - \mathbf{r}$ has only nonnegative entries
and $\mathbf{r}' \neq \mathbf{r}.$
Then we have:

\textbf{Observation 3.} \textit{If $\ket{\psi}$
admits a decomposition with PRV $\mathbf{r}$ and
$\ket{\phi}$ can be reached from $\ket{\psi}$ via SLOCC,
then $\ket{\phi}$ must also admit a decomposition with 
PRV $\mathbf{r}.$ Specifically, if $\ket{\phi}$ has a tight PRV $\mathbf{r}'$ and $\ket{\psi}$ admits a decomposition with $\mathbf{r} < \mathbf{r}'$,
then the SLOCC transformation from $\ket{\psi}$ to
$\ket{\phi}$  is impossible.}

To give an example, consider the state $\ket{\mathrm{GHZS}}$ for which the PRV $\mathbf{r}'= (2,2,0)$ is tight. This can not be reached via SLOCC from $\ket{\mathrm{GHZ}}$ since the GHZ state allows for a PRV
$\mathbf{r} = (2,1,0).$ It can also not be reached from
the singlet state $\ket{\mathrm{S}}$, since this state
allows for $\mathbf{r} = (1,1,1)$ but the $\ket{\mathrm{GHZS}}$ state does not. Note that these conclusions are not possible by looking at the Schmidt-number vector or the partition rank alone.

%%%%%%%%%%%%%%%%%%%%%%5
{\it Connections to mathematics.---}
%%%%%%%%%%%%%%%%%%%%%%
Let us finally discuss how our work and the perspective of quantum physics can help to understand problems arising in the field of mathematics. 

First, the method to study the computation of the partition rank via a PPT relaxation and a semidefinite
program is clearly physically motivated. This method allows proving upper bounds on the partition rank of a tensor. Our methods can also be adapted and applied to the tensor rank -- a decomposition into only fully separable states. In Appendix~\ref{app:borderrank}, we show this on the example of the 
matrix multiplication tensor~\cite{chitambar2008tripartite, ChristandlZuiddam2019matmultensor, ChristandlLysikovSteffanWernerWitteveen2024matmultensor,Strassen1969matmul}.

Second, the Open Problem 10 of Ref.~\cite{Shpilka2010openproblems}, formulated in terms of the partition rank in Ref.~\cite{Naslund2020} (Problem 12) has a nice physical interpretation. In 
the language of quantum mechanics, a stronger version of it reads as follows:

Consider two copies of a maximally entangled state $\ket{\Phi^+} =\nicefrac{1}{\sqrt{d}}\sum_{i=0}^{d-1} \ket{ii}$ in the  tensor product $\ket{\Phi^+}_{AB}\ket{\Phi^+}_{CD}$, how many terms are needed in a decomposition allowing only partitions $\{AC|BD,AD|BC\}$? How does this scale with the local dimension $d$?

Interestingly, for qubits ($d=2$), the decomposition 
\begin{align}
\label{eqRR}
\ket{\Phi^+}_{\mathrm{AB}}\ket{\Phi^+}_{\mathrm{CD}} = 
     \ket{\Phi^+}_{\mathrm{AC}}\ket{\Phi^+}_{\mathrm{BD}}+\ket{\Psi^-}_{\mathrm{AD}}\ket{\Psi^-}_{\mathrm{BC}},
\end{align}
where $\ket{\Psi^-}$ denotes a two-qubit singlet state,
is known since the early days of quantum mechanics; in fact, this inspired Georg Rumer (a student of Max Born) to develop a
nice graphical rule to write down a basis 
for a space spanned by multiqubit singlet state~\cite{rumer1932theorie}.

For arbitrary $d$, we obtain an upper bound on the number of terms by a generalization. We replace the last term in Eq.~(\ref{eqRR}) by
\begin{equation}
\ket{\Psi^-}_{\mathrm{AD}}\ket{\Psi^-}_{\mathrm{BC}}  \rightarrow \frac{1}{d}\sum_{0 \leq i<j<d} (\ket{ij}-\ket{ji})_{AD}\otimes(\ket{ij}-\ket{ji})_{BC} ,
\end{equation}
leading to the bound
\begin{equation}
    r_{\{AC|BD,AD|BC\}}(\ket{\Phi^+}_{AB}\ket{\Phi^+}_{CD})\le 1+d(d-1)/2.
\end{equation}
Still, this leaves the question of a bound with linear scaling in the dimension $d$ as suggested in Ref.~\cite{Naslund2020} 
open. Operationally, such decompositions are relevant for the construction of quantum networks, as they demonstrate that entanglement can be created between some node by entangling 
and superposing others.

{\it Conclusion.---}
We introduced a quantifier for the dimensionality of  genuine multipartite entanglement based on the partition rank. We connected it to well-known concepts, such as the geometric measure of entanglement and semi-definite programming relaxations, providing practical methods for its computation. With this we showed that our approach unveils properties of quantum states that were inaccessible to previous parameters. We finally showed that we can refine the partition rank using the partition rank vector. Our methods will find applications in the experimental characterization of multiparticle entanglement.

Our work opens new research perspectives as it provides novel tools to the quantum information community to tackle the highly-intricate problem of high-dimensional multipartite entanglement characterization. It also  may
stimulate the collaboration between physics 
and mathematics and computer science, to deepen the understanding of partition rank so that new applications in different areas of physics can be developed. In these directions, we provided a rephrasing of an open problem in quantum-information-theoretic terms. Future works may solve this or the more general version 
(Problem 13 in Ref.~\cite{Naslund2020}): Taking a state that is a tensor product of GHZ states along one partition, what is its partition rank along different sets of partitions?

\begin{acknowledgments}
    We thank M. Christandl, G. Cobucci, K. Goodenough, F. Otto, A. Rico, A. Tavakoli, L. T. Weinbrenner and J. Zuiddam for useful discussions. 
    This research was funded by the Deutsche Forschungsgemeinschaft (DFG, German Research Foundation, project numbers 563437167 and 572827488), the Sino-German Center for Research Promotion (Project M-0294), and the German Federal Ministry of Research, Technology and Space (Project QuKuK, Grant 
No.\ 16KIS1618K and Project BeRyQC, Grant No.\ 13N17292). SD acknowledges support from the House of Young Talents of the University of Siegen, and RK thanks the Hans and Ria Messer Stiftung. 

\textit{AI statement:} ChatGPT-5.6 Sol helped with the proofs in Appendix \ref{app:orth}.

\end{acknowledgments}

%%%%%%%%%%%%%%%%%%%%%%%%%%%%%%%%%%%%%%%%%%%%%%%%%%%
%%%%%%%%%%%%%%%%%%%%%%%%%%%%%%%%%%%%%%%%%%%%%%%%%%%
%%%%%%%%%%%%%%%%%%%%%%%%%%%%%%%%%%%%%%%%%%%%%%%%%%%
%%%%%%%%%%%%%%%%%%%%%%%%%%%%%%%%%%%%%%%%%%%%%%%%%%%
%%%%%%%%%%%%%%%%%%%%%%%%%%%%%%%%%%%%%%%%%%%%%%%%%%%
%%%%%%%%%%%%%%%%%%%%%%%%%%%%%%%%%%%%%%%%%%%%%%%%%%%
%%%%%%%%%%%%%%%%%%%%%%%%%%%%%%%%%%%%%%%%%%%%%%%%%%%
%%%%%%%%%%%%%%%%%%%%%%%%%%%%%%%%%%%%%%%%%%%%%%%%%%%
\onecolumngrid
\appendix
\renewcommand{\thesection}{\Alph{section}}

\section{Orthogonality of the decompositions}
\label{app:orth}
In this appendix we will discuss the orthogonality of partition rank decompositions. We will show that in the three-partite case, any partition
rank decomposition can be orthogonalized, which directly generalizes the Schmidt decomposition to the the three-particle case. We will then see that for $N$ particles an orthogonalization is not necessarily always possible.

Using the construction in Ref.~\cite{tao2016slice}, we note the following fact:
\begin{observation}\label{obs:orth}
    For any three-partite state in a (not necessarily optimal) partition rank decomposition $\ket{\psi} = \sum_{P\in\mathcal{P}}\sum_{\ell=1}^{\rr_P}\ket{\mu_\ell^P}$ the terms 
    $\ket{\mu_\ell^P}$ can be chosen to be mutually orthogonal. 
\end{observation}
\begin{proof}
Let us consider the decomposition
\begin{align}
        \ket{\psi} &= \sum_{P\in\mathcal{P}}\sum_{\ell=1}^{\rr_P}\ket{\mu_\ell^P}\\
         &= \ket{\xi_{A|BC}} +\ket{\xi_{B|AC}}+\ket{\xi_{C|AB}}\\
        &= \sum_{i=1}^{\rr_{A|BC}}\ket{\varphi^{A}_i}_A\ket{\phi^{BC}_i}_{BC}
        +\sum_{i=1}^{\rr_{B|AC}}\ket{\varphi^{B}_i}_B\ket{\phi^{AC}_i}_{AC}
        +\sum_{i=1}^{\rr_{C|AB}}\ket{\varphi^{C}_i}_C\ket{\phi^{AB}_i}_{AB}
\end{align}
According to Ref.~\cite{tao2016slice}, there are always $\rr_P$-dimensional subspaces, with projectors $\Pi_X$, for $X=A,B,C$, such that their orthogonal complements $\Pi_X^\perp$ fulfill
\begin{align}
    \Pi^\perp_{A}\otimes\Pi^\perp_{B}\otimes\Pi_C^\perp\ket{\psi} =0.
\end{align}
Now let us iteratively decompose the identity using that $\mathds{1}_X=\Pi_X+\Pi_X^\perp$:
\begin{align}
    \mathds{1} &= \Pi_A\otimes\mathds{1}_B\otimes\mathds{1}_C+\Pi_A^\perp\otimes\mathds{1}_B\otimes\mathds{1}_C\\
    &= \Pi_A\otimes\mathds{1}_B\otimes\mathds{1}_C
    +\Pi_A^\perp\otimes\Pi_B\otimes\mathds{1}_C+
    \Pi_A^\perp\otimes\Pi^\perp_B\otimes\mathds{1}_C\\
    &= \Pi_A\otimes\mathds{1}_B\otimes\mathds{1}_C
    +\Pi_A^\perp\otimes\Pi_B\otimes\mathds{1}_C+
    \Pi_A^\perp\otimes\Pi^\perp_B\otimes\Pi_C+
    \Pi_A^\perp\otimes\Pi^\perp_B\otimes\Pi_C^\perp
\end{align}
This expression will now help us to orthogonalize an arbitrary partition rank decomposition of a three-partite state. We have
\begin{align}
    \ket{\psi} = \mathds{1}\ket{\psi} = &\Pi_A\otimes\mathds{1}_B\otimes\mathds{1}_C\ket{\psi}\\
    +&\Pi_A^\perp\otimes\Pi_B\otimes\mathds{1}_C\ket{\psi}\\
    +&\Pi_A^\perp\otimes\Pi^\perp_B\otimes\Pi_C\ket{\psi}\\
    +&\Pi_A^\perp\otimes\Pi^\perp_B\otimes\Pi_C^\perp\ket{\psi}\\
    =& \sum_{i=1}^{\rr_{A|BC}}\ket{\varphi^{A}_i}_A\ket{\phi^{BC}_i}_{BC}+
           \sum_{i=1}^{\rr_{B|AC}}\ket{\varphi^{B}_i}_B\Pi_A\otimes\mathds{1}_C\ket{\phi^{AC}_i}_{AC}
        +\sum_{i=1}^{\rr_{C|AB}}\ket{\varphi^{C}_i}_C\Pi_A\otimes\mathds{1}_B\ket{\phi^{AB}_i}_{AB}\label{eq:A}\\
    +& \sum_{i=1}^{\rr_{B|AC}}\ket{\varphi^{B}_i}_B\Pi^\perp_A\otimes\mathds{1}_C\ket{\phi^{AC}_i}_{AC}
        +\sum_{i=1}^{\rr_{C|AB}}\ket{\varphi^{C}_i}_C\Pi_A^\perp\otimes\Pi_B\ket{\phi^{AB}_i}_{AB}\label{eq:B}\\
    +&\sum_{i=1}^{\rr_{C|AB}}\ket{\varphi^{C}_i}_C\Pi_A^\perp\otimes\Pi_B^\perp\ket{\phi^{AB}_i}_{AB}  \label{eq:C} \\
    =&\ket{\xi'_{A|BC}}+\ket{\xi'_{B|AC}}+\ket{\xi'_{C|AB}}.
\end{align}
In the last step, we summarized all terms in Line~\eqref{eq:A} as $\ket{\xi'_{A|BC}}$, the terms in Line~\eqref{eq:B} as $\ket{\xi'_{B|AC}}$
and in Line~\eqref{eq:C} as $\ket{\xi'_{C|AB}}$.  
All terms are mutually orthogonal, which is 
ensured by the respective distribution of 
projectors $\Pi_X$ and $\Pi_X^\perp$ projecting 
onto orthogonal subspaces. By construction the $\ket{\xi'_P}$ have Schmidt rank at most $\rrp$ 
since the respective projectors fix the maximal dimension of the respective single-party subspace 
$X$. Writing down their Schmidt decomposition gives the
desired partition-rank decomposition with all vectors
orthogonal.
\end{proof}

This observation has two interesting implications. 
First, the proof can be extended in complete analogy to the $N$ particle case. Important is however, that we may only consider terms that are separable with respect to a 1|rest partition, i. e. we only consider the slice rank. In this case, we know from Ref.~\cite{tao2016slice}
that projectors $\Pi_X$ and $\Pi_X^\perp$ with $\Pi_{A_1}^\perp\otimes\cdots\otimes\Pi_{A_N}^\perp \ket{\psi}=0$ always exist. Then we can proceed as in the three-particle case and decompose the identity into projectors onto the respective subspaces and their orthogonal. We will obtain $N+1$ terms, where the last term only consist of projectors onto the orthogonal subspaces and therefore vanishes when applying to $\ket{\psi}$. All other $N$ terms will build the vectors $\ket{\xi'_P}$ for each of the $N$ considered partitions $P$ and the orthogonality is given by construction as in the three-particle case.

Second, for any other set of partitions in the $N>3$-particle case, we can derive constraints for the non-existence of an orthogonal decomposition. 
\begin{observation}
    Any $N$-particle state $\ket{\psi}$ whose Schmidt coefficients $s_i^P$ with respect to partition $P$ are ordered decreasingly, fulfills 
    \begin{align}
        \mathcal{S}=\sqrt{\sum_{P\in\mathcal{P}}\sum_{i=1}^{\rr_{P}}\left(s_i^P\right)^2}\geq 1
        \label{eq-kriterium}
    \end{align}
    if an orthogonal partition rank decomposition into $k = \sum_{P\in\mathcal{P}}\rrp$ terms corresponding to the partition rank vector $\mathbf{r}=(\rrp)_{P\in \mathcal{P}}$ is possible. Hence, if the Schmidt coefficients do not fulfill this condition, such a decomposition into orthogonal terms cannot exist.
\end{observation}
\begin{proof}
    Let $\ket{\psi} = \sum_{P\in\mathcal{P}}\ket{\xi_P} = \sum_{P\in\mathcal{P}}\sum_{\ell=1}^{\rr_P}\ket{\mu_\ell^P}$ be an orthogonal decomposition of an $N$ particle state into $k$ terms, with partition rank vector $\mathbf{r} = (\rr_{P_i})_i$. Then, for its norm must hold:
    \begin{align}
       1 &= |\braket{\psi}{\psi}| = 
       |\bra{\psi}\sum_{P\in\mathcal{P}}\ket{\xi_P}|
       \leq |\alpha_1|\sqrt{\sum_{i=1}^{\rr_{P_1}} \left(s^{P_1}_i\right)^2}+|\alpha_2|\sqrt{\sum_{i=1}^{\rr_{P_2}} \left(s^{P_2}_i\right)^2}+\cdots \\
       &\leq \sqrt{\sum_{P\in\mathcal{P}}\sum_{i=1}^{\rr_{P}}\left(s_i^P\right)^2} \sqrt{\sum_j|\alpha_j|^2}
       =  \sqrt{\sum_{P\in\mathcal{P}}\sum_{i=1}^{\rr_{P}}\left(s_i^P\right)^2} := \mathcal{S}. 
    \end{align}
    In the first step, we used the triangle inequality and the fact that the squared overlap of a state $\ket{\psi}$ with a state of Schmidt rank $r$ is upper bounded by the sum of the $r$ largest squared Schmidt coefficients of $\ket{\psi}$. The prefactors $|\alpha_i|$ denote the norm of the $\ket{\xi_P}$. In the second step we used the Cauchy-Schwarz inequality and lastly, since we assumed that all terms are orthogonal we have $\sqrt{\sum_j|\alpha_j|^2}=1$.
\end{proof}

We can use this criterion to give an example for $N=4$ particles of a decomposition which can not be orthogonalized. Consider  the state
\begin{align}
\ket{\psi} = 
     \ket{\Phi^+}_{\mathrm{AC}}\ket{\Phi^+}_{\mathrm{BD}}+\ket{\Psi^-}_{\mathrm{AD}}\ket{\Psi^-}_{\mathrm{BC}}.
\end{align}
Here we have $\ket{\xi_{AC|BD}}=\ket{\Phi^+}_{\mathrm{AC}}\ket{\Phi^+}_{\mathrm{BD}}$ and $\ket{\xi_{AD|BC}}=\ket{\Psi^-}_{\mathrm{AD}}\ket{\Psi^-}_{\mathrm{BC}}$, which are not orthogonal. We obtain
\begin{align}
    \mathcal{S} = \sqrt{\left(s_1^{AC|BC}\right)^2+\left(s_1^{AD|BC}\right)^2} = \frac{1}{4}+\frac{1}{4} = \frac{1}{2} <1,
\end{align}
which violates the criterion. Hence, for this state an \textit{orthogonal} decomposition with respect to the partitions $AC|BC$ and $AD|BC$ is not possible.

Let us now consider the three-particle case again. Here we know that an orthogonal decomposition must always exist. Thus, we can use the criterion to make statements on the partition rank itself, or to be more precise on the partition rank vector. We take the states from Table~I in the main text. For the GHZ state and the Supersinglet state we directly see that no partition rank 2 decomposition is possible, since all Schmidt coefficients are equal to $\nicefrac{1}{\sqrt{3}}$, so for all possible combinations of partition rank two decompositions we obtain $\mathcal{S}=\nicefrac{2}{3}<1$. Additionally, for the state $\ket{\psi(1,1,1)}$ we find that it has Schmidt coefficients
$s_1^{A|BC} = \nicefrac{1}{\sqrt{2}}$, $s_2^{A|BC} = \nicefrac{1}{\sqrt{3}}$, $s_1^{B|AC}=s_2^{B|AC}=\nicefrac{1}{\sqrt{3}}$ and $s_1^{C|AB}=s_2^{C|AB}=\nicefrac{1}{\sqrt{3}}$, hence we have at most 
\begin{align}
    \mathcal{S}=\frac{1}{2}+\frac{1}{3} = \frac{5}{6} < 1,
\end{align}
ruling out a partition rank 2 decomposition. For the state $\ket{\psi(3,4,5)}$ we could not rule out a partition rank 2 decomposition using this criterion.

Let us give some final remarks. First, note that although this criterion
has similarities to fidelity-based entanglement criteria, it is only applicable to pure states. In order to detect partition rank $k$ experimentally, the quantitative criteria derived in the main text (or a fidelity bound derived from Eq.~(\ref{eq-kriterium}) ) have to be applied. Second, particularly in the three-partite case, the criterion shows well that high entanglement across several partitions implies high partition rank. This is because if a state $\ket{\psi}$ has high Schmidt rank across all bipartitions, its Schmidt coefficients are low and it will violate the criterion which implies a high partition rank.
Third, although the criterion in the tripartite case can in principle exclude certain partition rank vectors, it can only make statements on partition rank vectors corresponding to a decomposition with $k$ terms if the local dimension is at least $d=k+1$. We can see this directly since for any partition rank vector, the smallest possible value for $\mathcal{S}$ is reached by AME states, whose Schmidt coefficients are all equal to $\nicefrac{1}{\sqrt{d}}$. We then have
\begin{align}
     \mathcal{S} = \sum_{i=1}^k\frac{1}{d} = \frac{k}{d}, 
\end{align}
which is only smaller than one if $d>k$.

\section{Methods for the computation of the geometric measure of partition rank $k$}
\label{app:algo}
In this appendix we describe the technical details of the algorithms that allow us to approximate a given state $\ket{\psi}$ with a state $\ket{\eta}$ of a fixed partition rank. We start with the two approaches that generalize the geometric measure of entanglement. Then we will discuss the method based 
on the Gauß-Newton algorithm, which is particularly useful for finding large partition rank $k$ approximations. 

\subsection{Semidefinite program}
First, we will reformulate our problem in terms of the geometric measure, which is also useful for the see-saw algorithm we will discuss later.
Next we will explain the relaxation which allows us to use semidefinite programming.
\subsubsection{Reformulation in terms of the geometric measure}
For a given pure quantum state $\ket{\psi}$, we want to maximize the 
following expression
\begin{align}
\Omega^2 = \sup_{\ket{\eta}}|\braket{\psi}{\eta}|^2,
\label{eq-appendix:overlap}
\end{align}
where $\ket{\eta}$ has a specific structure with respect to the partition rank. We now fix the structure as $\ket{\eta} = \ket{\varphi_1}_A\ket{\phi_1}_{BC}+\ket{\phi_2}_{AB}\ket{\varphi_2}_{C}$, where $\ket{\eta}$ is a normalized state. The algorithm can be analogously adapted to any other partition rank structure of $\ket{\eta}$.

In order to formulate Eq.~\eqref{eq-appendix:overlap} as an SDP, we first need to express $\ket{\eta}$ as an unnormalized product state. To that end, we define a matrix operator $E$, such that
\begin{align}
   \ket{\eta} = E\ket{s}\ket{L},
\label{eq:Emat}
\end{align}
where $\ket{s} = \ket{\varphi_1}_A \oplus \ket{\varphi_2}_C$ 
and $\ket{L} = \ket{\phi_1}_{BC} \oplus \ket{\phi_2}_{AB}$. 

Let us explicitly construct the matrix $E$ such that Eq.~\eqref{eq:Emat} 
is fulfilled. To that end, we think of $E$ as a map that takes the 
$2d \times 2d^2$-dimensional vector $\ket{s}\ket{L}$, living on a
tensor product of direct sums of spaces and maps it to a 
$d^3$-dimensional vector $\ket{\eta}$, living on the actual 
Hilbert space.  We define two projectors $\Pi_{1}$ and $\Pi_2$ 
which allow us to project onto the first term and second term 
of the decomposition, 
\begin{equation}
    \ket{\varphi_1}_A \ket{\phi_1}_{BC} \propto \Pi_1 |s\rangle |L\rangle \quad \mathrm{ and } \quad \ket{\phi_2}_{AB} \ket{\varphi_2}_C \propto \Pi_{2} |s\rangle |L\rangle.
\end{equation}
In the following we will first construct these projectors onto the space of 
dimension $2d \times 2d^2 = 4d^3$; second, we choose an embedding 
into the $d^3$-dimensional physical Hilbert space.

The first projection $\Pi_{1}$ can be written as:
\begin{equation}
    \Pi_{1} = Q^{\mathrm{s}}_1\otimes Q^{\mathrm{L}}_1,
\end{equation}
where $Q^{s/L}_1$ denotes projection on the first term $\ket{\varphi_1}_A \ket{\phi_1}_{BC}$, and the indices $s$ and $L$ refer to the vectors $\ket{s}$ and $\ket{L}$. In this case, we differentiate between a small system, containing the one-party vectors, and the large system, containing the two-party vectors. The projections are then given by,

\begin{equation}
    Q^{\mathrm{s}}_1 = \begin{pmatrix}
    \mathds{1}_{d} & \mathbb{O}_{d} \\
    \mathbb{O}_{d} & \mathbb{O}_{d}
    \end{pmatrix}, \qquad Q^{\mathrm{L}}_1 = \begin{pmatrix}
    \mathds{1}_{d^2} & \mathbb{O}_{d^2} \\
    \mathbb{O}_{d^2} & \mathbb{O}_{d^2}
    \end{pmatrix},
\end{equation}
where $\mathds{1}_d$ is the $d\times d$ identity matrix and $\mathbb{O}_d$ the $d\times d$ zero matrix.

For the second projector $\Pi_{2}$, we can define the projection operators similarly to before 
\begin{equation}
    Q^{\mathrm{s}}_2 = \begin{pmatrix}
    \mathbb{O}_{d} & \mathbb{O}_{d} \\
    \mathbb{O}_{d} & \mathds{1}_{d}
    \end{pmatrix}, \qquad Q^{\mathrm{L}}_2 = \begin{pmatrix}
    \mathbb{O}_{d^2} & \mathbb{O}_{d^2} \\
    \mathbb{O}_{d^2} & \mathds{1}_{d^2}
    \end{pmatrix}.
\end{equation}

Apart from applying these operators we additionally need to permute 
the subsystems, such that they are in the correct order $ABC$. And further, 
we have to ensure that the coefficients are in the right order as well, 
namely, following the notation according to the computational basis, the 
first entry in the vector should be mapped to $\ket{000}$, the second 
entry to $\ket{001}$ and so on. This can reached by a permutation $\pi$ 
of the $4d^3$  rows of $Q^{\mathrm{s}}_2 \otimes Q^{\mathrm{L}}_2$ 
and we obtain
\begin{equation}
\Pi_{2} = M_\pi \circ  [Q^{\mathrm{s}}_2 \otimes Q^{\mathrm{L}}_2],
\end{equation}
where $M_\pi$ denotes the permutation matrix. Note that this permutation
makes the operator $\Pi_{2}$ a nonlocal operator with respect to the 
tensor product structure $\ket{s}\ket{L}.$

Lastly, in order to ensure that $E$ maps from a $4d^3$-dimensional space 
to a $d^3$-dimensional space, we embed the operators accordingly. This 
is expressed by applying the operator $\Gamma$, which essentially 
eliminates the rows of the operators $\Pi_1$ and $\Pi_2$ that contain 
only zeros,
\begin{equation}
    E = \Gamma \left ( \Pi_1 + \Pi_2 \right ).
\end{equation}
The total operator $E$ is then a $d^3 \times 4d^3$ matrix and we can 
reformulate
\begin{align}
    \sup_{\ket{\eta}} |\braket{\psi}{\eta}|^2 &= \sup_{\ket{s}\ket{L}}|\bra{\psi}E\ket{sL}|^2\\
    &= \sup_{\ket{s}\ket{L}}|\bra{\Psi\ket{sL}}^2 ,
\end{align}
where $\bra{\Psi} = \bra{\psi}E$. Note that we have the normalization constraint
$\braket{\eta}{\eta}=\bra{sL}E^\dagger E \ket{sL}=1$ here.

\subsubsection{SDP relaxation}
Now we can relax the condition that $\ketbra{sL}{sL}$ is a separable state, to the condition that it has a positive partial transpose and replace $\ketbra{sL}{sL}\mapsto \sigma$. Then
\begin{equation}
    \sup_{\ket{s}\ket{L}}|\bra{\Psi\ket{sL}}^2\leq \sup_{\sigma_\mathrm{PPT}}\tr(\ketbra{\Psi}{\Psi}\sigma_{\mathrm{PPT}}).
\end{equation}
This allows us to formulate the problem as an SDP:
\begin{align}
    \sup_{\sigma} &\tr(\ketbra{\Psi}{\Psi}\sigma)\\
    \mbox{subject to: } &\sigma\quad \mathrm{positive~semidefinite},\\ 
    & \sigma \quad \mathrm{PPT},\\   &\tr (E\sigma E^\dagger)=1,\label{eq:normalSDP}\\   
    & \tr ( \sigma ) = 2\label{eq:normSDP2},
\end{align}
where  the condition in Eq.~\eqref{eq:normalSDP} encodes the normalization, which distinguishes this program from the one computing bounds on the geometric measure of entanglement and will be recovered in the see-saw approach. The second normalization condition in Eq.~\eqref{eq:normSDP2} does not directly follow from our reformulation but is crucial for the PPT relaxation. Precisely, since we map the state $\ket{\eta}$ from
a $d^3$-dimensional space to a the way larger direct sum space of dimension $4d^3$, performing a PPT relaxation without additional constraint might cause a bad approximation.
However, demanding $\tr(\sigma)=2$ in the direct sum space does not add any contraints to the standard space, where $\ket{\eta}$ lives. This is because as shown in Appendix~\ref{app:orth}, any partition rank decomposition of a three-partite state can be
orthogonalized. Then we can for example choose the $\ket{\varphi_i}$, which are the entries of the lifted vector $\ket{s}$ to be of norm one, while the $\ket{\phi_i}$ will have norm $\alpha_i$. We then obtain
\begin{align}
    \tr(\sigma) = \braket{s,L}{s,L}= \sum_i \norm{\ket{\varphi_i}}^2\sum_i \norm{\ket{\phi_i}}^2 = 2(|\alpha_1|^2+|\alpha_2|^2)=2.
\end{align}
Note that for any $\ket{\eta}$ of partition rank $k$, or slice rank $k$ for $N$ particles,
we can impose the constraint $\tr(\sigma)=k$ without restricting the actual set of $\ket{\eta}$, we are optimizing over. For arbitrary partition rank of $N>3$ particles, it remains open to find a similar suitable condition.

\subsection{See-Saw algorithm}
We now consider another method, based on a see-saw algorithm and show how this is related to the SDP approach. Again, we want to maximize
\begin{align}
    \Omega^2 = \sup_{\ket{\eta}}|\braket{\psi}{\eta}|^2,
\end{align}
and fix the structure of $\ket{\eta}$ to $\ket{\eta} = \ket{\varphi_1}_A\ket{\phi_1}_{BC}+\ket{\phi_2}_{AB}\ket{\varphi_2}_{C}$. 

As before, we can rewrite the problem by rearranging the parts of the decomposition of $\ket{\eta}$ into two vectors  $\ket{s} = \ket{\varphi_1}_A \oplus \ket{\varphi_2}_C$ 
and $\ket{L} = \ket{\phi_1}_{BC} \oplus \ket{\phi_2}_{AB}$.  The idea of the see-saw algorithm is now to first fix the vector $\ket{L}$ and find the optimal $\ket{s}$ and then fix $\ket{s}$ and find the optimal $\ket{L}$. This can then
be iterated.

Let us fix $\ket{L}$. In order to rewrite the maximization problem $\sup_{\ket{\eta}}|\braket{\psi}{\eta}|^2 
{=}\sup_{\ket{sL}}|\braket{\Psi}{sL}|^2$  
as a see-saw optimization, we first compute 
$\bra{V} = \braket{\Psi}{L}$ and obtain $\bra{V} = \bra{v_1} \oplus \bra{v_2}$.
At first sight, a natural choice is now to take $\ket{s} \sim \ket{V}$,
but we have to take the special normalization into account. 

Since $\ket{\eta}$ is normalized but the vectors $\ket{\varphi_i}$ 
and $\ket{\phi_i}$ (and therefore $\ket{s}=\ket{s_1}\oplus \ket{s_2}$ 
and $\ket{L}$) are not, the normalization condition can be written as
\begin{align}
    \braket{\eta}{\eta} = &\bra{\varphi_1}_A\bra{\phi_1}_{BC}\ket{\phi_1}_{BC}\ket{\varphi_1}_A + \bra{\varphi_1}_A\bra{\phi_1}_{BC}\ket{\phi_2}_{AB}\ket{\varphi_2}_C %\nonumber \\
    %&
    +\bra{\varphi_2}_C\bra{\phi_2}_{AB}\ket{\phi_1}_{BC}\ket{\varphi_1}_A +\bra{\varphi_2}_C\bra{\phi_2}_{AB}\ket{\phi_2}_{AB}\ket{\varphi_2}_C 
    \nonumber \\
    = &\sum_{ij} \bra{s_i}M_{ij}\ket{s_j} =1,
\end{align}
where $M$ is an operator-valued $2 \times 2 $ matrix given by
\begin{align}
    M = \left(
    \begin{array}{cc}
       \bra{\phi_1}_{BC}\ket{\phi_1}_{BC} \times \mathds{1}_A & \bra{\phi_1}_{BC}\ket{\phi_2}_{AB} \\
       \bra{\phi_2}_{AB}\ket{\phi_1}_{BC}  & \bra{\phi_2}_{AB}\ket{\phi_2}_{AB} \times \mathds{1}_C 
    \end{array}
    \right),
\end{align}
where each entry is a $d\times d$ block and obtained by contracting 
the vectors $\ket{\phi_i}$ with respect to the appropriate subsystems. 

So, we have to solve the maximization problem
\begin{align}
    \sup_{\ket{s}} |\braket{V}{s}|,\,\,\text{subject to}\,\,\bra{s}M\ket{s}=1.
\end{align}
Note that this would be solved by taking $\ket{s} \sim \ket{V}$ in the case
of $M=\mathds{1}$, which is exactly the see-saw algorithm for the geometric measure of entanglement.

For the case $M\neq \mathds{1}$, we can use the Lagrange multiplier formalism to identify the optimal vector $\ket{s}$. With the Lagrange function
\begin{align}
    \mathcal{L} = \braket{V}{s}-\lambda (\bra{s}M\ket{s} -1)
\end{align}
we obtain the result \cite{lagrage1,lagrange2}
\begin{align}
    \lambda^2 &= \bra{V}M^{-1}\ket{V}\\
    \bra{s} &= \frac{\bra{V}M^{-1}}{\lambda}.
\end{align}
Hence, the optimal vector $\ket{s}$ is given by the hermitian conjugate
\begin{align}
    \ket{s} = \Big(\frac{\bra{V}M^{-1}}{\sqrt{\bra{V}M^{-1}\ket{V}}}\Big)^\dagger
    =
     \Big(\frac{M^{-1} \ket{V}}{\sqrt{\bra{V}M^{-1}\ket{V}}}\Big),
    \label{eq:pvec}
\end{align}
since $M$ is hermitean. 
Then the first $d$ entries of $\ket{s}$ define the updated vector $\ket{\varphi_1}_A = \ket{s_1}$ and the next $d$ entries $\ket{\varphi_2}_C= \ket{s_2}$.
Next, we repeat the procedure fixing the $\ket{\varphi_i}$, i. e. the vector $\ket{s}$ and updating the $\ket{\phi_i}$ (the vector $\ket{L}$).

Note that for non-singular $M$ this is always well-defined as the expectation values of $M$ are real and nonzero. When updating the larger system, the matrix $M$ becomes singular and therefore is not invertible. We can overcome this problem by using the pseudoinverse $M^+$ instead. The problem is then well-defined, if $\ket{V}$ is in the range of $M$, since then $M^+ \ket{V} = M^+MM^+\ket{V} = M^+ \Pi_\mathrm{range}\ket{V}\neq 0$. We note that in practice this is always the case, since $\ket{V}$ is a generic vector.

\subsection{Perturbation-theoretic algorithm}
We introduce an algorithm for finding partition rank $k$ decompositions, 
based on the Gauss-Newton-type method, which supplements the previously 
introduced see-saw method and can deliver lower bounds on $\Omega^2$.
We perform the optimization over decompositions  
in the following way.

We take one ansatz state $\ket{\gamma_P}$ for each partition $P\in\mathcal{P}$,
and optimize these states, under the constraint
$\sum_P \ket{\gamma_P} = \ket{\eta}$. {Said 
constraint is implemented by singling out a 
specific partition $P^*$ and setting $\ket{\gamma_{P^*}} = \ket{\eta}-\sum_{P\neq P^*}\ket{\gamma_P}$, as an implementation by parametrization saves resources in the numerical implementation.} The partition $P^*$ varies during
the iteration and in each step the other partitions 
are updated in the following way.

If the partition $P$ can 
be written as $P = M | \overline M$ we consider the Schmidt decomposition
$\ket{\gamma_P} = \sum_i \sigma^{P}_i \ket{u^{M}_i} \otimes \ket{v^{\overline M}_i}$ across the partition $P = M | \overline M$. {Note that this decomposition might initially not satisfy the required partition rank vector, but will do so asymptotically, if the PRV is possible for the state.}
The goal is to modify $\ket{\gamma_P}$ such that all but $\rr_P$ of these 
Schmidt coefficients are zero, where $\rr_P$ is the Schmidt rank of the 
contribution in the target state $\ket{\eta}$; in the case given above we 
aim to reach $\rr_{A|BC}=1$, $\rr_{B|AC}=0$ and $\rr_{C|AB}=1$. Note that the 
approach uses a natural bipartite factorization of the Hilbert space into 
$\mathcal{H}_{M}\otimes\mathcal{H}_{\overline M}$ and hence the number of components in $\sigma^P_i$ is 
$d^{\min}_P =\min(\dim(\mathcal{H}_{M}),\dim(\mathcal{H}_{\overline M}))$.
From the fixed structure of $\ket{\eta}$, we know how many $\sigma^{P}_i$ across each bipartition ought to be zero for a valid decomposition, concretely the smallest $d^{\min}_{P}-\rr_P$ of them.

We collect all the  $\ket{\gamma_P}$ in a large vector,
$\ket{\bm{\gamma}} = \bigoplus_P\ket{\gamma_P}$, and collect
also all of the Schmidt coefficients that should vanish in 
a vector,
\begin{equation}
    \bm{\sigma}:=\bigoplus_{P\in \mathcal{P},i>\rr_P} \sigma^P_i. 
\end{equation}
Suppose now that all components of $\bm{\sigma}$ take on distinct 
values. In this specific case, we can 
use first-order perturbation theory to calculate the derivative: 
\begin{equation}
    \bm{J}:= \frac{\partial\bm{\sigma}}{\partial\bra{\bm{\gamma}}} = \bigoplus_P\frac{\partial\bm{\sigma}}{\partial\bra{\gamma_P}} = \bigoplus_{P\in \mathcal{P},i>\rr_P} \bra{u^{P_{1}}_i}\otimes\bra{v^{P_{2}}_i}
\end{equation}
We can then compute updates $\ket{\bm{\gamma}}\mapsto\ket{\bm{\gamma}}+\ket{\bm{\delta}}$ 
by solving
\begin{equation}    \bm{J}\ket{\bm{\delta}} = -\bm{\sigma},
\end{equation}
resulting in a type of Gauss-Newton method. Similar ideas have been presented in Ref.~\cite{magne2010numerical}. 

\section{Optimal and tight partition rank decompositions}
\label{app:dec}
In this appendix we will give some details on optimal partition rank decompositions. Precisely, we will show
that if a state is in its optimal partition rank decomposition, then it can be interpreted as a sum of states which are written in
the Schmidt decomposition with respect to the partition that minimizes the Schmidt rank. 
The entires of the optimal partition rank vector, will then correspond to the GME-dimension of these states.
We will then see that this is not a sufficient condition for a partition rank vector (and its assigned decomposition) to be optimal,
which motivates the definition of tight PRVs.
\begin{lemma}
Any state with partition rank $r_\mathcal{P}$ can be written in the form
\begin{align}
    \ket{\psi} = \sum_{P\in \mathcal{P}} \ket{\xi_P},
\end{align}
where each unnormalized state $\ket{\xi_\mathrm{P}}$ has partition rank equal to its minimial Schmidt rank (GME-dimension). The index $P$ denotes the respective partition in which the decomposition of the $\ket{\xi_\mathrm{P}}$ is optimal. Then, the partition rank is given by the sum of the Schmidt ranks $\rr_P$ $(=d_\mathrm{GME})$ with respect to bipartition $P$ of the terms $\ket{\xi_\mathrm{P}}$
\begin{align}
    r_\mathcal{P}(\ket{\psi}) = \sum_{P\in \mathcal{P}} \rr_P(\ket{\xi_P}) =\sum_{P\in \mathcal{P}} d_{\mathrm{GME}}( \ket{\xi_P})
\end{align}
and the assigned partition rank vector is optimal with entires $\rrp =d_\mathrm{GME}(\ket{\xi_P})$.
\label{th:declem}
\end{lemma}
\begin{proof}
    Consider a state $\ket{\psi}$ which has partition rank $r_\mathcal{P} = \sum_{P\in \mathcal{P}} \rr_P(\ket{\xi_P})$. We can decompose the state into a sum of $r_\mathcal{P}$ biseparable states and summarize all $\rr_P$ states that are biseparable with respect to the same partition $P$
    \begin{align}
        \ket{\psi} = \sum_{P\in\mathcal{P}}\Big( \sum_{i=1}^{\rr_P}\ket{\nu_{i}^P}\Big)
        = \sum_{P\in \mathcal{P}} \ket{\xi_P}.
    \end{align}
    The terms $\ket{\nu_{i}^P}$ can be orthogonalized by the Gram-Schmidt method, hence each state $\ket{\xi_P}$ has Schmidt rank $\rr_P$ with respect to the bipartition $P$. We further find by contradiction that $\rr_P(\ket{\xi_P})=d_\mathrm{GME}(\ket{\xi_P})$ must hold. Namely, if the state $\ket{\xi_P}$ had a smaller Schmidt rank with respect to another bipartition we could express it with less terms and therefore the whole state $\ket{\psi}$ would have a partition rank smaller than $r_\mathcal{P}$. Since we assumed that $\ket{\psi}$ has patition rank $r_\mathcal{P}$ we find that indeed all $\ket{\xi_P}$ have minimal Schmidt rank with respect to bipartition $P$ and therefore  
    \begin{align}
    r_\mathcal{P}(\ket{\psi}) = \sum_{P\in \mathcal{P}} \rr_P(\ket{\xi_P}) =\sum_{P\in \mathcal{P}} d_{\mathrm{GME}}(\ket{\xi_P}).
\end{align} 
\end{proof}
Note that Lemma~\ref{th:declem} does not imply the converse: If we have a state of the form $\ket{\psi} = \sum_{P\in \mathcal{P}}\ket{\gamma_P}$, we cannot say that its partition rank is given by the sum of the GME-dimensions of the $\ket{\gamma_P}$. Consider for example the four-qubit state
\begin{align}
\ket{\psi}=\ket{\Phi^+}_{\mathrm{AB}}\ket{\Phi^+}_{\mathrm{CD}} = 
     \ket{\Phi^+}_{\mathrm{AC}}\ket{\Phi^+}_{\mathrm{BD}}+\ket{\Psi^-}_{\mathrm{AD}}\ket{\Psi^-}_{\mathrm{BC}},
\end{align}
where $\ket{\Psi^-}\propto \ket{01}-\ket{10}$ and $\ket{\Phi^+}\propto \ket{00}+\ket{11}$. Both terms are biseparable with respect to a different bipartition and their GME dimension is equal to the Schmidt rank in that bipartition, so one might assume that $\ket{\psi}$ has partition rank 2. However, in fact $\ket{\psi}$ is biseparable with respect to the bipartition $AC|BD$ and has therefore partition rank 1. Indeed, the decomposition on the right hand side is tight with the assigned PRV $\mathbf{r}=(0,0,0,0,0,1,1)$ is tight
but not optimal, whereas the decomposition on the left-hand side, following the PRV $\mathbf{r}=(0,0,0,0,1,0,0)$ is optimal.

\section{Algebraic certification}
\label{app:cert}
In the following, we explain an algebraic-geometric method to find lower bounds on the partition rank of a fixed $\ket{\psi}$. 
Concretely, we can also lower-bound partition rank vectors associated to decompositions of $\ket{\psi}$.
For that, we use Hilbert's Nullstellensatz, which links inconsistency or unsatisfiability of a collection of polynomial equations to the existence of a solution to a linear system of equations over polynomials. Concretely, a collection of polynomial equations generates an ideal, and if this ideal contains the constant polynomial, the equations are inconsistent. The method we describe here can be viewed as a generalization of the approach in Ref.~\cite{johnston2022complete}.

First, we assign coordinates $\bm{\gamma}:=(\gamma^{(P)}_{jk})_{j,k,P}$ to a set of vectors in a potential decomposition, $\ket{\gamma_P} := \sum_{jk}\gamma^{(P)}_{jk}\ket{jk}$. Then we define $g_0(\bm{\gamma}) = \sum_P\ket{\gamma_P} - \ket{\psi}$. 
Note that this is inhomogeneous in $\bm{\gamma}$, circumventing the issue of normalization. That is because if we were to allow decompositions of the 
entire ray $\alpha\ket{\psi}$, then we would always find that $\ket{\gamma}=0$ 
is a decomposition for $\alpha=0$.

We then gather the set of all minors $\{g^{(P)}_{a_P}(\bm{\gamma})\}_{P, 1\le a_P \le N^c_P}$ corresponding to a partition rank vector $\bm{r}$. This means
that, for every fixed $P$, the determinants of sub-blocks with size $(\rr_P+1)$ in the matrix $\gamma^{(P)}_{jk}$ are considered and enumerated with an 
index $a_P$, where $N^c_P$ is the number of such minors. The variety $\{\bm{\gamma} | g^{(P)}_{a_P}(\bm{\gamma})=0\}$  contains 
decompositions with partition rank vector $\leq \bm{r}$. This means
that any such decomposition of $\ket{\psi}$ solves all of the homogeneous equations $\{g_a(\bm{\gamma})=0\}$, including $g_0$, so that the set of decompositions of a state has the correct Schmidt rank constraints imposed. 

However, we can also consider the polynomials generated by assigning to every partition one specific
minor $g^{(P)}_{a_P}$ and multiplying them, $\tilde{g}_a = \prod_P g^{(P)}_{a_P}$. The associated variety $\{\bm{\gamma} | \tilde{g}_a(\bm{\gamma})=0\}$ also contains all decompositions $\not\geq\bm{r}$, since such expressions are only nonzero if all factors are nonzero, which is only the case if all Schmidt ranks of $\ket{\gamma_P}$ are incremented with respect to the test vector $\bm{r}$. Alternatively, varieties of a given partition rank vector can be considered directly, by not multiplying the minors $g^{(P)}_{a_P}$ as above, but inserting them directly as equations. This imposes Schmidt rank constraints on the individual terms $\ket{\gamma_P}$. 

Now, instead of proving that a variety $\{\bm{x}|G_a(\bm{x})=0\}$ is empty, we can equivalently show that $\{\sum_a G_a(\bm{x})Q_a(\bm{x})|Q_a(\bm{x})\in \mathbb{C}[\bm{x}]\}$, the ideal generated by the $G_a$, contains the constant polynomial. Here, $\mathbb{C}[\bm{x}]$ is the ring of all polynomials with complex coefficients in variables $\bm{x}$. 
That is, we attempt to find $Q_a(\bm{\gamma})$ 
solving 
\begin{equation}
    \sum_{a=0}^{N_c}  Q_a(\bm{\gamma})G_a(\bm{\gamma}) =1.
\end{equation}
Essentially, for such an inconsistency certificate, one performs simplification operations on a set of polynomial constraints, to show that one of the constraints is "$1=0$", which, is clearly not satisfiable. Expanding the above expression in the monomial basis of bounded degree gives rise to a hierarchy of linear equations where any solution certifies the emptiness of the variety generated by the aforementioned polynomials. Thus, we can obtain lower bounds on partition rank vectors, and hence on the partition rank.

To demonstrate our method in practice, we now apply it to prove that no decomposition following the partition rank vector $\mathbf{r}=(1,1,1)$ exists for the state
\begin{equation}
    \ket{\mathrm{GHZS}}_{\mathrm{ABC}}=  \frac{1}{3}\sum_{i=0}^2\ket{iii}_{\mathrm{ABC}} + \frac{1}{3}\sum_{ijk}\epsilon_{ijk}\ket{ijk}_{\mathrm{ABC}}.
\end{equation}
To this end, we set up the associated system of degree-$2$ minors and multiply it with the set of monomials of degree $\le2$. Now we search for a linear combination of these polynomials that will result in the constant polynomial. 
One solve this system of linear equations with floating point numbers, and then remove all columns whose coefficients are close to zero. Repeating this considerably lowers the size of the linear system. Finally, we randomly select columns and set their coefficients to zero, while preserving consistency of the system of linear equations. This finally results in a highly sparse solution which can be rounded to a rational solution. Such a rational solution can be exactly verified in contrast to a floating point solution, proving that the set of decompositions following the vector $\mathbf{r}=(1,1,1)$ of $\ket{\mathrm{GHZS}}_{\mathrm{ABC}}$ is empty. Details and programs for this 
calculation are available upon request.

We note that the above method cannot distinguish decompositions defined as a limit of finite decompositions (that is, border partition rank), but not associated to any decomposition with finite terms summing up to the normalized state. A different method can be used for excluding well-behaved decompositions: Here, we assign to every partition in the partition rank vector with nonzero Schmidt rank a quadratic term $\ket{\gamma_P} = \sum_{i=0}^{\rr_P-1}\ket{a^{(P)}_i}_{P_1}\otimes\ket{b^{(P)}_i}_{P_2}$, then we require $\sum_P\ket{\gamma_P}=\ket{\psi}$, which gives a system of quadratic equations in $\ket{a^{(P)}_i}_{P_1}$, $\ket{b^{(P)}_i}_{P_2}$. If we now use the Nullstellensatz to acquire a certificate of inconsistency for these equations it will not include the diverging limit points, since now the set belonging to the partition rank vector is parametrized directly, without imposing constraints from the outside, such that the case of $\ket{\gamma_P}$ diverging is not equated to a solution. As an illustration, consider the equations $xy=1$ and $y=0$, which have the pathological solutions $(x,y)=(\pm\infty,0)$, if $x=\nicefrac{1}{\pm y}$. Regardless, $(x)(y) + (-1)(xy-1)=1$ shows that $1$ is contained in the ideal generated by the two equations, such that they are inconsistent.

%%%%%%%%%%%%%%%%%%%%%%%%%%%%%%%%%%%%%%%%%%%
\section{Proof of Observation 2}
\label{app:AME}
%%%%%%%%%%%%%%%%%%%%%%%%%%%%%%%%%%%%%%%%%%%%
We now want to prove Observation 2 from the main text, which states that for three particles, the geometric measure of partition rank 2 is maximized if and only if $\ket{\psi}$ is an absolutely maximally entangled state. To do so, we first give an alternative, more explicit, proof of the fact that all partition rank two decompositions of three particles can be orthogonalized.
\begin{lemma}
\label{th:orthon}
    Every partition rank two state on three parties has an orthogonal decomposition, i.e. $\ket{\psi}_{ABC} = \ket{\nu_1^{P_1}}+\ket{\nu_2^{P_2}}$, such that $\braket{\nu_1^{P_1}}{\nu_2^{P_2}}=0$, where $\ket{\nu_i^{P_i}}$ are unnormalized states and separable with respect to partition $P_i$.
\end{lemma}
\begin{proof}
If the two states $\ket{\nu_1}$
and $\ket{\nu_2}$ are biseparable for the same bipartition, the 
statement follows from the properties of
the standard Schmidt decomposition. Otherwise, we consider a three-partite state of partition rank two which, without loss of generality, can be written as
\begin{equation}
    \ket{\psi}_{ABC} = \ket{\nu_1^{A|BC}} + \ket{\nu_2^{B|AC}} = \ket{a}_A\ket{\phi_1}_{BC} + \ket{b}_B\ket{\phi_2}_{AC}.
\end{equation}
We can now add a zero term without changing the state,
\begin{align*}
    \ket{\psi}_{ABC} &= \ket{a}_A\ket{\phi_1}_{BC} + \ket{b}_B\ket{\phi_2}_{AC} + \ket{a}_A\ket{b}_B\ket{c}_C - \ket{a}_A\ket{b}_B\ket{c}_C\\
     &= \ket{a}_A \left (\ket{\phi_1}_{BC}+\ket{b}_B\ket{c}_C \right ) + \ket{b}_B \left(\ket{\phi_2}_{AC} -  \ket{a}_A\ket{c}_C \right )\\
     :&= \ket{\tilde{\nu}_1^{A|BC}} + \ket{\tilde{\nu}_2^{B|AC}}.
\end{align*}
Now, we can say that the state $\ket \psi_{ABC}$ admits an orthogonal decomposition if and only if there is a $\ket{c}_C$ such that $\langle \tilde{\nu}_1|\tilde{\nu}_2 \rangle=0$, where we ommit the superscript for better readability. Then, we have 
\begin{align}
    \langle \tilde{\nu}_1|\tilde{\nu}_2 \rangle &= \bra b \langle \phi_2 | a \rangle \ket \phi_1 + \langle b | b\rangle_B\langle \phi_2 | ac\rangle_{AC}  - \langle a |a\rangle_A \langle b c| \phi_1 \rangle_{BC} - \langle abc | abc \rangle  .
\end{align}
We can factorize this expression as (using $\mathcal{N} = \langle ab|ab\rangle$),
\begin{align}
     \langle \tilde{\nu}_1|\tilde{\nu}_2 \rangle &= \left(-\sqrt{\mathcal{N}}\bra{c}_C + \frac{1}{\sqrt{\mathcal{N}}}\bra{b}_B\bra{\phi_2}_{AC}\ket{ab}_{AB}\right)\left(\sqrt{\mathcal{N}}\ket{c}_C + \frac{1}{\sqrt{\mathcal{N}}}\langle a | a \rangle_A\bra{b}_B\ket{\phi_1}_{BC}\right).
\end{align}
As stated previously, $\ket \psi_{ABC}$ admits an orthogonal decomposition if and only if there is a $\ket{c}_C$ such that $\langle \tilde{\nu}_1|\tilde{\nu}_2 \rangle=0$. 
If already $\langle \nu_1|\nu_2 \rangle=0$ then $|c\rangle _C = 0$ . Otherwise, there are two solutions for $|c\rangle _C $ that allow transforming any non-orthogonal partition-rank-2 decomposition into an orthogonal one,
\begin{align}
    |c\rangle _C = \frac{1}{\mathcal N}\langle b|b\rangle_B\langle a |_A |\phi_2 \rangle_{AC} \qquad \mathrm{or} \qquad |c\rangle _C = -\frac{1}{\mathcal N}\langle a|a\rangle_B\langle b |_B |\phi_2 \rangle_{BC}
\end{align}
\end{proof}

We can now prove Observation 2 from the main text.\\

\textbf{Observation 2.} \textit{For three-particle states, the geometric measure of partition rank two is maximized by $\ket{\psi}$ if and only if $\ket{\psi}$ is absolutely maximally entangled. We then have $\Omega^2 = \sup_{\eta\in\mathbb{P}_2}|\braket{\eta}{\psi}|^2 = \nicefrac{2}{d}$.}

\begin{proof}
First, note that the overlap with all states $\ket{\eta}$ of Schmidt rank two, is given by the sum of the squared largest two Schmidt coefficients. It can therefore not be smaller than $\nicefrac{2}{d}$ and this value is only reached if the largest two Schmidt coefficients are given by $\nicefrac{1}{\sqrt{d}}$, which is the case only for absolutely maximally entangled states.

We now have to show that the overlap with states of partition rank two, but not Schmidt rank two cannot exceed the bound $\nicefrac{2}{d}$. Due to the symmetry of the state family, we can assume without loss of generality that the closest partition rank two state is of the form $\ket{\eta} = \ket{\nu_1^{A|BC}}+\ket{\nu_1^{B|AC}} =\alpha \ket{\varphi_1^{A|BC}}+\beta\ket{\varphi_2^{B|AC}}$, with 
$\ket{\varphi_1^{A|BC}}$ and $\ket{\varphi_2^{B|AC}}$ normalized. Consider the maximization
    \begin{align}
        &\sup_{\alpha,\beta,\ket{\varphi_1^{A|BC}},\ket{\varphi_2^{B|AC}}}|\bra{\psi}(\alpha \ket{\varphi_1^{A|BC}}+\beta\ket{\varphi_2^{B|AC}})|^2 \\
        &\leq \sup_{\alpha,\beta,\ket{\varphi_1^{A|BC}},\ket{\varphi_2^{B|AC}}}\left(\frac{|\alpha|^2}{d}+\frac{|\beta|^2}{d}+2\Re(\alpha^*\beta\braket{\varphi_1^{A|BC}}{\psi}\braket{\psi}{\varphi_2^{B|AC}})\right)\\
        &\leq\sup_{\alpha,\beta,\ket{\varphi_1^{A|BC}},\ket{\varphi_2^{B|AC}}}\left(\frac{1}{d}\left(|\alpha|^2+|\beta|^2+2|\alpha||\beta|\right)\right) = \sup_{\alpha,\beta}\frac{1}{d}(|\alpha| +|\beta|)^2.
    \end{align}
    In the first step we used that $\ket{\psi}$ is AME, so the largest overlap with all biseparable states is $\nicefrac{1}{d}$. In the second step, we used 
    that $\Re(ab)\leq|a||b|$ and the bound $|\braket{\psi}{\varphi_1^{A|BC}}|\leq\nicefrac{1}{\sqrt{d}}$ that holds for AME states.

    Now, we only have to maximize the expression $(|\alpha|+|\beta|)^2$ under the constraint that the state $\ket{\eta} = \alpha \ket{\varphi_1^{A|BC}}+\beta\ket{\varphi_2^{B|AC}}$ is normalized. By Lemma~\ref{th:orthon}, we can assume that $\ket{\eta}$ is in an orthonormal decomposition and therefore the normalization condition is given by $|\alpha |^2 +|\beta|^2=1$. Hence, the expression is maximized when $\alpha=\nicefrac{1}{\sqrt{2}}=\beta$ reaching the value 2.
    This yields the total overlap $\Omega^2 = \sup_{\eta\in\mathbb{P}_2}|\braket{\eta}{\psi}|^2 = \nicefrac{2}{d}$, which proves the observation.
 
\end{proof}

\section{Analytical bounds on the geometric measure of slice rank two}
\label{app:slice}
We note that since the proof of Lemma~\ref{th:orthon} is independent of the dimension, the statement generalizes to any number of parties. Note again, that this statement is already included in Observation~\ref{obs:orth} and we here just give an additional proof for completeness. This allows us to give bounds on the geometric measure of slice rank two.

\begin{cor}
    Every slice rank two state has an orthogonal decomposition.
\end{cor}
\begin{proof}
    This follows directly from the proof of Lemma~\ref{th:orthon} and the fact that it is independent of the dimension. Consider
    \begin{align}
        \ket{\psi}_{A_1\cdots A_N} &= \ket{a}_{A_1}\ket{\chi_1}_{A_2\cdots A_N}+\ket{b}_{A_2}\ket{\chi_2}_{A_1A_3\cdots A_N}  \nonumber \\
        &= \ket{a}_{A_1}\ket{\chi_1}_{A_2\cdots A_N}+\ket{b}_{A_2}\ket{\chi_2}_{A_1A_3\cdots A_N} + \ket{a}_{A_1}\ket{b}_{A_2}\ket{c}_{A_3\cdots A_N}-\ket{a}_{A_1}\ket{b}_{A_2}\ket{c}_{A_3\cdots A_N}
    \end{align}
    we can now interpret $\ket{c}_{A_3\cdots A_N}$ as a single party, living on a larger system and perfom the orthogonalization in analogy to the proof of Lemma~\ref{th:orthon}.
\end{proof}
Then we can continue computing a bound on the overlap $\Omega_*^2 = \sup_{\kappa\in \mathbb{S}_2}|\braket{\kappa}{\psi}|^2$, where the index $*$ denotes that we maximize over slice rank two states and not partition rank two states. We compute
    \begin{align}
        &\sup_{\alpha,\beta,\ket{\varphi_1^{A_1|A_2\cdots A_N}},\ket{\varphi_2^{A_2|A_1A_3\cdots A_N}}}|\bra{\psi}(\alpha \ket{\varphi_1^{A_1|A_2\cdots A_N}}+\beta\ket{\varphi_2^{A_2|A_1A_3\cdots A_N}})|^2 \\
        &\leq \sup_{\alpha,\beta,\ket{\varphi_1^{A_1|A_2\cdots A_N}},\ket{\varphi_2^{A_2|A_1A_3\cdots A_N}}}\left(s_{A_1}^2|\alpha|^2+s_{A_2}^2|\beta|^2+2\Re(\alpha^*\beta\braket{\varphi_1^{A_1|A_2\cdots A_N}}{\psi}\braket{\psi}{\varphi_2^{A_2|A_1A_3\cdots A_N}})\right)\\
        &\leq\sup_{\alpha,\beta}\left(s_{A_1}^2|\alpha|^2+s_{A_2}^2|\beta|^2+2|\alpha||\beta|s_{A_1}s_{A_2}\right) 
        = \sup_{\theta}\left(s_{A_1}^2\sin^2(\theta)+s_{A_2}^2\cos^2(\theta)+2\sin(\theta)\cos(\theta)s_{A_1}s_{A_2}\right).
    \end{align}
The steps are in analogy to the AME case in Appendix~\ref{app:AME}. In the more general case we consider here, we obtain a dependency on $s_{A_i}$, which are the largest Schmidt coefficients of $\ket{\psi}$ taken with respect to the bipartition $A_i|A_1\cdots A_{i-1}A_{i+1}\cdots A_N$. The last equation then gives an upper bound for the overlap with slice rank two states. 

In Fig.~\ref{fig:anabound} we plot the non-trivial bounds as a function of the Schmidt coefficients $s_{A_i}$ and $s_{A_j}$. Note, that the bounds themselves do not depend on the local dimension, but the possible combinations of Schmidt coefficients grows with the dimension. 
\begin{figure}[htbp]
    \centering
    \includegraphics[width=0.5\linewidth]{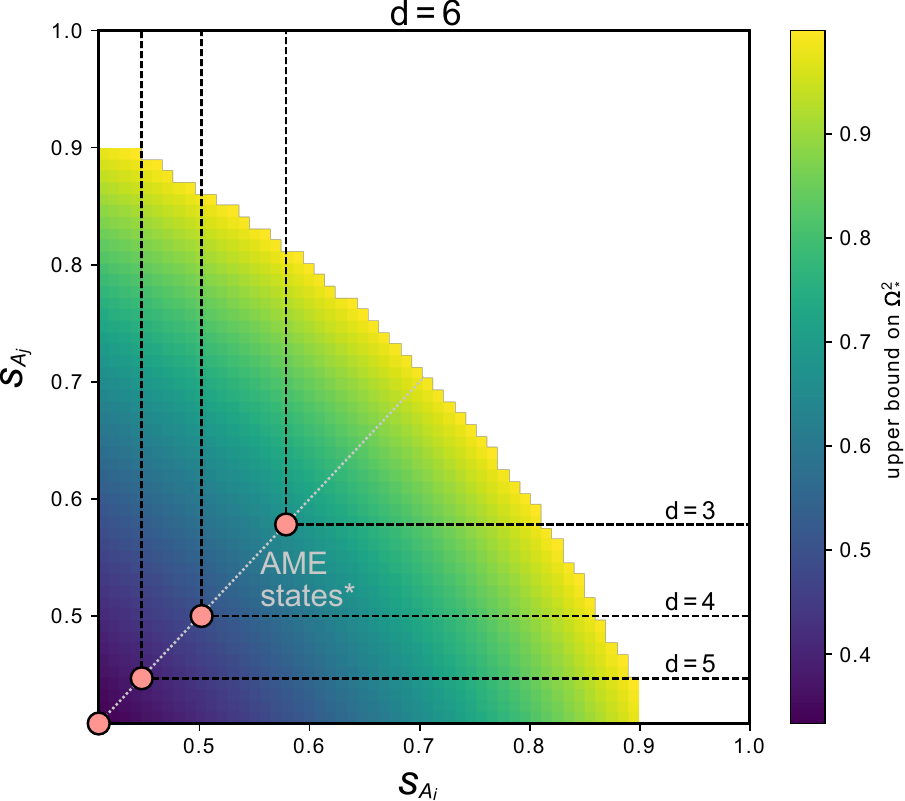}
    \caption{Analytical bound for the overlap with slice rank two states depending only on the Schmidt coefficients. The coefficient $s_{A_i}$ is the largest Schmidt coefficient of $\ket{\psi}$ when considering the cut $A_i|$rest and analog for $A_j$. This plot was done for local dimension $d=6$, restricting the possible values for the Schmidt coefficients; the black dashed lines show the smaller dimensions $d\in\{3,4,5\}$. In the three particle case the orange points correspond to AME states. 
    For the white area, the upper bound is trivial.}
    \label{fig:anabound}
\end{figure}

\section{Extensions to tensor rank}
\label{app:borderrank}
First, recall that a state that can be decomposed as $\ket{\tau}=\sum_{i=1}^{k}\ket{t_i}$, with no less than $k$ fully separable terms $\ket{t_i}$,
has tensor rank $k$. For a state $\ket{\psi}$, that fulfills
\begin{align}
    \sup_{\ket{\tau}\in \mathbb{T}_k}|\langle\psi |\tau\rangle|^2=1,
\end{align}
where the supremum is taken over all tensor rank $k$ states, we say that it has \textit{border rank} $k$. Note that tensor rank and border rank do not necessarily coincide.
A simple example is given by the W state $\ket{W}\propto\ket{010}+\ket{001}+\ket{100}$, which has tensor rank 3, but border rank 2~\cite{deSilvaillposed,KlinglerNetzerDelesCoves2025border}.
 
Let us now show that our SDP and see-saw approach extend to the tensor rank, allowing systematic derivations of tensor rank decompositions and provable upper and lower bounds on the border rank.

The logic is the same as above: Assume we want to approximate a state $\ket{\psi}$ with a state $\ket{\tau}$ of tensor rank 2:
\begin{align}
     \ket{\tau} = \ket{\varphi^A_1}_A\ket{\varphi^B_1}_B\ket{\varphi^C_1}_C+\ket{\varphi_2^A}_A\ket{\varphi_2^B}_B\ket{\varphi_2^C}_C
\end{align}
As in the partition rank case, we collect the states corresponding to specific subsystems in a vector by taking the direct sum. Here we define $\ket{a} = \ket{\varphi_1^A}_A \oplus \ket{\varphi_2^A}_A$, $\ket{b} = \ket{\varphi_1^B}_B \oplus\ket{\varphi_2^B}_B$ and $\ket{c} = \ket{\varphi_1^C}_C \oplus \ket{\varphi_2^C}_C$ and rewrite:
\begin{align}
    \sup_{\ket{\tau}\in \mathbb{T}_2} |\braket{\psi}{\tau}|^2 \mapsto \sup_{\ket{abc}}|\braket{\Psi}{abc}|^2,\,\,\text{s. th. } \bra{abc}E^\dagger_TE_T\ket{abc} = 1,
\end{align}
where $\bra{\Psi} = \bra{\psi}E_T$ and $E_T$ is an operator lifting $\ket{\tau}$ into the direct sum subspace: $\ket{\tau} = E_T\ket{abc}$.
Then we can continue as in the partition rank case and perform a see-saw or solve an SDP relaxation.

There are some comments in order: First, note that now we have to compute the overlap with a multipartite fully separable state instead of a bipartite one. The methods still work in an analogous way, but for the see-saw we now have to perform three updates, $\ket{a}$, $\ket{b}$ and $\ket{c}$. Moreover, for the SDP, the PPT relaxation has to be implemented by setting $\ket{abc}\bra{abc}\mapsto \sigma$, where $\sigma$ has a positive partial transpose with respect to all subsystems. 

Second, the number of particles in the fully separable state scales with the number of particles of the state $\ket{\psi}$, as for each system $X$, we define a new vector $\ket{x}$. The dimension of these vectors scales with the local dimension of the state $\ket{\psi}$ multiplied by the tensor rank, since the tensor rank of $\ket{\tau}$ gives the number of terms in the direct sum. 

Third, note that although in the end we optimize over a product state, we can still only make conclusions on the border rank, not on the tensor rank. Indeed, computing the overlap $\sup_{\ket{\tau}\in \mathbb{T}_2}|\braket{W}{\tau}|^2$ we see that the see-saw algorithm slowly converges until it reaches a one in machine precision.

Lastly, we can give upper and lower bounds on the border rank and therefore also lower bounds on the tensor rank. The logic here is the following: The SDP gives a provable upper bound on the overlap $\sup_{\ket{\tau}\in \mathbb{T}_k}|\braket{\psi}{\tau}|^2$. This means that if we obtain an overlap smaller than one for some tensor rank $k_1$, and an overlap one for $k_2 = k_1+1$, we know that the state $\ket{\psi}$ must have at least border rank $k_1+1$. For the see-saw, we know that if the overlap is smaller than one for some $k_1'$ and one for $k_2'=k_1'+1$ the upper bound on the border rank is $k_2'$.

We can use this to verify the border rank of the matrix multiplication tensor, which for $n\times n$ matrices is given by \cite{chitambar2008tripartite, ChristandlZuiddam2019matmultensor, ChristandlLysikovSteffanWernerWitteveen2024matmultensor}:
\begin{align}
    \ket{\psi_n} = \sum_{ijk=0}^{n-1} \ket{ij}_A\otimes\ket{jk}_B\otimes \ket{ki}_C.
\end{align}
For $n=2$ this can be interpreted as a $4\times4\times4$-system and in this case it is known that it has tensor rank and border rank equal to 7 \cite{Strassen1969matmul, Landsberg2006bordermatmul}. We verify this result with our see-saw algorithm, providing a systematic method for the identification of tensor rank decompositions.

We want to test the limits of our see-saw algorithm and go further by defining $\ket{\psi_{n,n,l}}$ as above, but now the sum over $i$ and $j$ runs until $n-1$, whereas the sum over $k$ runs until $l-1$. This is the matrix multiplication tensor for multiplying $n\times n$ matrices with $n\times l$ matrices~\cite{Landsbergmatmul}.

Choosing $n=2$ and $l=3$, the state $\ket{\psi_{2,2,3}}$ becomes a $4\times 6 \times 6$-tensor and it is known that it has border rank 10 \cite{borderrank10}. When approximating $\ket{\psi_{n,n,l}}$ with a state $\ket{\tau}$ of tensor rank 10, the algorithm reaches an overlap of 0.9998 until the numerical error becomes to large to ensure the normalization of the approximated state $\ket{\tau}$. {Since we have to compute a pseudoinverse in each step, a numerical error is collected and shows up in the normalization of $\ket{\tau}$, we therefore stop the algorithm when $|\braket{\tau}{\tau} - 1| \geq 10^{-8}$}.

To our best knowledge, the first case where the exact border rank is unknown is for $n=2$ and $l = 4$. Here $\ket{\psi_{2,2,4}}$ is a $4\times 8\times 8$-tensor and the border rank lies between 12 and 13~\cite{landsberg2015geometryborderrankalgorithms}. Since the see-saw algorithm gives an upper bound, finding an approximation with a state $\ket{\tau}$ of tensor rank 12 would verify border rank 12 for this matrix multiplication tensor. Our algorithm reaches an overlap of 0.8749 with a tensor rank 12 state after 30000 iterations and an overlap of 0.9959 with a tensor rank 13 state, which indicates (but does not prove) that the actual border rank might be 13 and not 12.

\twocolumngrid

\bibliography{refs}

\end{document}